%% file: main.tex
\documentclass{sig-alternate-10}

\input{packages}
\input{macros}

\input{title_authors}

\begin{document}
\maketitle
\sloppy

\input{abstract}

\section{Introduction}\label{sec:intro}
\input{intro}

\section{Basic definitions}\label{sec:prelim}
\input{prelim}

\section{Constant delay evaluation of extended Vset automata}\label{sec:algorithm}
\input{cde-algorithm}

\section{Evaluating regular spanners}\label{sec:spanners}
\input{spanners1}

\section{Counting document spanners}\label{sec:lower}
\input{lowerbounds}

\section{Conclusions}\label{sec:concl}
\input{conclusions}

%%%%%%%%%%%%%%%%%%%
%%% Bibliography %%
%%%%%%%%%%%%%%%%%%%

\newpage
\bibliographystyle{plain}
\bibliography{biblio}

\newpage
\onecolumn
\appendix

\newcommand{\atitle}{\ref{sec:algorithm}}
\section{Proofs from Section~\atitle}\label{sec:appendix}
\input{appendix-algorithm}

\renewcommand{\atitle}{\ref{sec:spanners}}
\section{Proofs from Section~\atitle}\label{sec:appendix-spanner}
\input{appendix-spanners}

\renewcommand{\atitle}{\ref{sec:lower}}
\section{Proofs from Section~\atitle}\label{sec:appendix-lower}
\input{appendix-spanl}

\end{document}

%% file: packages.tex
\usepackage{amsmath}
\usepackage{amssymb}
\usepackage[utf8]{inputenc}
\usepackage{bm}
\usepackage{mathabx}
\usepackage{mathtools}
\usepackage{stmaryrd}
\usepackage{algorithm}
\usepackage[noend]{algpseudocode}
\usepackage{comment}
\usepackage[textwidth=2cm,textsize=small]{todonotes}
\usepackage{paralist} %for compact enumeration enviornments
\usepackage{tikz}
\usetikzlibrary{automata, graphs,positioning,chains,arrows,decorations.pathmorphing}
\usepackage{xspace}

%% file: macros.tex
\newcommand{\plist}{\operatorname{\textit{list}}}
\newcommand{\plistold}{\plist^{\operatorname{\textit{old}}}}

\renewcommand{\dag}{\textnormal{DAG}\xspace}

\def\mspan#1#2{[#1,#2\rangle}
{\def\n#1{\textsf{\tiny{#1}}}
\def\ccell#1{\multicolumn{1}{c}{#1}}
\def\s#1{\mathtt{#1}}

\newcommand{\algops}{\ensuremath{{\{\pi,\cup,\Join\}}}}

\newcommand{\newnode}{\operatorname{node}}
\newcommand{\newmap}{\operatorname{map}}

\newcommand{\markers}{\operatorname{Markers}}

\newcommand{\gname}{\ensuremath{\gamma_{\text{n}}}}
\newcommand{\gemail}{\ensuremath{\gamma_{\text{e}}}}
\newcommand{\gphone}{\ensuremath{\gamma_{\text{p}}}}

\usepackage{enumitem}
\newlist{todolist}{itemize}{2}
\setlist[todolist]{label=\rlap{$\Box$}{\phantom{\cmark}}}
\usepackage{pifont}
\newcommand{\cmark}{\ding{51}}%
\newtheorem{theorem}{Theorem}[section]
\newtheorem{example}[theorem]{Example}

\newtheorem{proposition}[theorem]{Proposition}

\newtheorem{lemma}[theorem]{Lemma}

\newcommand{\bbN}{\mathbb{N}}

\newcommand{\cA}{\mathcal{A}}
\newcommand{\cL}{\mathcal{L}}

\newcommand{	\VV}{\mathcal{V}}

\newcommand{	\sub}{\text{span}}

\newcommand{	\var}{\text{var}}

\newcommand{	\dom}{\text{dom}}
\newcommand{\sem}[1]{\llbracket#1\rrbracket}

\newcommand{\semp}[1]{[#1]}

\newcommand{\semd}[1]{\sem{#1}_{d}}
\newcommand{\seme}[2]{\sem{#1}_{#2}}

\newcommand{\sempd}[1]{\semp{#1}_{d}}

\newcommand{\trans}[2][]{\raisebox{-1pt}[10pt][0pt]{$\overset{#2}{\underset{^{#1}}{\raisebox{0pt}[3pt][0pt]{$\relbar\mspace{-8mu}\rightarrow$}}}$}}

\newcommand{\VA}{\ensuremath{\mathrm{VA}}}
\newcommand{\EVA}{\ensuremath{\mathrm{eVA}}}
\newcommand{\sEVA}{\ensuremath{\mathrm{seVA}}}
\newcommand{\fEVA}{\ensuremath{\mathrm{feVA}}}
\newcommand{\dsEVA}{\ensuremath{\mathrm{\text{deterministic }seVA}}}

\newcommand{\fVA}{\ensuremath{\mathrm{fVA}}}
\newcommand{\sVA}{\ensuremath{\mathrm{sVA}}}

\newcommand{\RGX}{\ensuremath{\mathrm{RGX}}}

\newcommand{\Var}{\operatorname{var}}

\newcommand{\Open}[1]{#1~\mkern-10mu\vdash}
\newcommand{\Close}[1]{\dashv~\mkern-10mu#1}

\newcommand{\ENUM}{\text{\sc Enumerate}}

\newcommand{\COUNT}{\text{\sc Count}}

\newcommand{\LL}{{\cal L}}

\newcommand{\Out}{\text{\sc Out}}

%% file: title_authors.tex
\title{Constant delay algorithms for regular document spanners}

\numberofauthors{5}

\author{\alignauthor
Fernando Florenzano \\
\affaddr{PUC Chile}\\
\email{faflorenzano@uc.cl}
\and
\alignauthor
Cristian Riveros \\
\affaddr{PUC Chile}\\
\email{cristian.riveros@uc.cl}
\and
\alignauthor
Mart\'in Ugarte\\
\affaddr{Universit\'e Libre de Bruxelles}\\
\email{mugartec@ulb.ac.be}
\and
\alignauthor
Stijn Vansummeren\\
\affaddr{Universit\'e Libre de Bruxelles}\\
\email{stijn.vansummeren@ulb.ac.be}
\and
\alignauthor
Domagoj Vrgo\v{c}\\
\affaddr{PUC Chile}\\
\email{dvrgoc@ing.puc.cl}
}

%% file: abstract.tex
%!TEX root = ../main/main.tex
\begin{abstract}
  Regular expressions and automata models with capture variables are
  core tools in rule-based information extraction. These formalisms,
  also called regular document spanners, use regular languages in
  order to locate the data that a user wants to extract from a text
  document, and then store this data into variables.
  Since document spanners can easily generate large outputs, it is
  important to have good evaluation algorithms that can generate the
  extracted data in a quick succession, and with relatively little
  precomputation time. Towards this goal, we present a practical
  evaluation algorithm that allows constant delay enumeration of a
  spanner's output after a precomputation phase that is linear in the
  document. While the algorithm assumes that the spanner is specified
  in a syntactic variant of variable set automata, we also study how
  it can be applied when the spanner is specified by general variable
  set automata, regex formulas, or spanner algebras. Finally, we study
  the related problem of counting the number of outputs of a document
  spanner, providing a fine grained analysis of the classes of
  document spanners that support efficient enumeration of their
  results.
\end{abstract}

%%% Local Variables: 
%%% mode: latex
%%% TeX-master: "../main/main"
%%% End: 

%% file: intro.tex
%!TEX root = ../main/main.tex

%\stijn{I believe that there needs to be a comment made in the
%introduction that we focus on the class of \emph{regular} spanners
%(which is more limited than the core spanners. We had such a comment
%in section 2 previously, but I deleted this some time ago.}
%
%\stijn{Actually, does it make sense to rename the paper to ``CDE
%algorithms for *regular* documents spanners?}
%\domagoj{Will do in contributions.}

%\noindent{\bf Background.} 
Rule-based information extraction (IE for short)
\cite{ChiticariuLR13,FKRV15,Kimelfeld14} has received a fair amount of
attention from the database community recently, revealing interesting
connections with logic \cite{F17,freydenberger2016document}, automata
\cite{FKRV15,MaturanaRV17}, datalog programs \cite{AMRV16,SDNR07}, and
relational languages \cite{CKLRRV10,KLRRVZ08,FKP17}. In rule-based IE,
documents from which we extract the information are modelled as
strings. This is a natural assumption for many formats in use today
(e.g. JSON and XML files, CSV documents, or plain text). The extracted
data are represented by {\em spans}. These are intervals inside the
document string that record the start and end position of the
extracted data, plus the substring (the data) that this interval
spans. The process of information extraction can then be abstracted by
the notion of {\em document spanners} \cite{FKRV15}: operators that
map strings to tuples containing spans.
% operators that take strings as their input, and produce as their
% output a representation of the data extracted from these strings.

The most basic way of defining document spanners is to use some form
or regular expressions or automata with capture variables. The idea is that a regular language is used in order to locate the data to be extracted, and variables to store this data. This approach to IE has been widely adopted in the database literature \cite{FKRV15,FKRV14,AMRV16,F17,MaturanaRV17}, and also forms the core extraction mechanism of commercial IE tools such as IBM's SystemT \cite{KLRRVZ08}. The two classes of expressions and automata for extracting information most commonly used in the literature are {\em regex formulas} (RGX) and {\em variable-set automata} (\VA), both formally introduced in \cite{FKRV15}.

A crucial problem when working with $\RGX$ and $\VA$ in practice is
how to evaluate them efficiently. One issue here is that the output
can easily become huge. For the sake of illustration, consider the regex formula
$\gamma = \Sigma^*\cdot x_1\{\Sigma^*\cdot x_2\{\Sigma^*\} \cdot
\Sigma^*\}\cdot \Sigma^*,$ where $\Sigma$ denotes a finite
alphabet. Intuitively, $\gamma$ extracts any span of a document
$d$ into $x_1$, and any sub-span of this span into
$x_2$. Therefore, on a document $d$ over $\Sigma$ it will produce an
output of size $\Omega(|d|^2)$.  If we keep nesting the variables
(i.e., $x_3$ inside $x_2$, etc.), the output size will be
$\Omega(|d|^\ell)$, with $\ell$ the number of variables in $\gamma$.
Since an evaluation algorithm must at least write down this output,
and since the latter is exponential (in $\gamma$ and $d$),
alternative complexity measures need to be used in order to
answer when this problem is efficiently solvable.

A natural option here is to use {\em enumeration algorithms}
\cite{Segoufin13}, which work by first running a pre-computation
phase, after which they can start producing elements of the output
(tuples of spans in our case) with pre-defined regularity and without
repetitions. The time taken by an enumeration algorithm that has an
input $I$ and an output $O$ is then measured by a function that
depends both on the size of $I$ and the size of $O$. Ideally, we would
like an algorithm that runs in total time $O(f(|I|) + |O|)$, where $f$
is a function not depending on the size of the output, so that the
output is returned without taking much time between generating two of
its consecutive elements. This is achieved by the class of {\em
  constant delay enumeration algorithms} \cite{Segoufin13}, that do a
pre-computation phase that depends only on the size of the input
($\gamma$ and $d$ in our case), followed by an enumeration of the
output without repetitions where the time between two outputs is
constant.

Constant delay algorithms have been studied in various contexts,
ranging from MSO queries over trees
\cite{bagan2006mso,courcelle2009linear}, to relational conjunctive
queries \cite{Bagan:2007}. These studies, however, have been mostly
theoretical in nature, and did not consider practical applicability of
the proposed algorithms. To quote several recent surveys of the area:
``We stress that our study is from the theoretical point of view.  If
most of the algorithms we will mention here are linear in the size of
the database, the constant factors are often very big, making any
practical implementation difficult."
\cite{Segoufin13,segoufin2014glimpse,Segoufin15}. These surveys also
leave open the question of whether practical algorithms could be
designed in specific contexts, where the language being processed is
restricted in its expressive power. This was already shown to be true
in \cite{AMRV16}, where a constant delay enumeration algorithm for a
restricted class of document spanners known as navigation expressions
was implemented and tested in practice. Since navigation expressions
are a very restricted subclass of $\RGX$ and $\VA$
\cite{MaturanaRV17}, and since the latter have been established in the
literature as the two most important classes of rule-based IE
languages, in this paper we study practical constant delay 
algorithms for $\RGX$ and $\VA$.

%\stijn{Is the reference \cite{MaturanaRV17} the correct one ? I  thought that this was shown by Kimmelfeld et al in their upcoming  pods paper (but I probably misunderstood ...).}

\smallskip
\noindent{\bf Contributions.} The principal contribution of our work
is an intuitive constant delay algorithm for evaluating a syntactic
variant of $\VA$ that we call extended $\VA$. Extended $\VA$ are
designed to streamline the way $\VA$ process a string, and the algorithm we present can evaluate an extended $\VA$ $\cA$ that is both sequential \cite{MaturanaRV17} and deterministic over a document $d$ with
pre-processing time $O(|\cA|\times |d|)$, and with constant delay
output enumeration. We then study how this algorithm can be applied to
arbitrary $\RGX$ and $\VA$, and their most studied restrictions such
as functional and sequential $\RGX$ and $\VA$. Both sequential and
functional $\VA$ and $\RGX$ are important subclasses of regular spanners: as shown in \cite{FKRV15,MaturanaRV17,FKP17}, they have both good
algorithmic properties and prohibit unintuitive behaviour.  Next, we
proceed by extending our findings to the setting where spanners are
specified by means of an algebra that allows to combine $\VA$ or $\RGX$
using unions, joins and projections. 
As such, we identify  upper bounds on the preprocessing times when evaluating the class
of regular spanners \cite{FKRV15} with constant delay.

In an effort to get some idea of potential lower-bounds on preprocessing times, we study the problem of counting the number of tuples output by a spanner.
%Furthermore, we also study the problem of counting the number of tuples output by a spanner. 
This problem is strongly connected to the enumeration problem
\cite{Segoufin13}, and gives evidence on whether a constant delay
algorithm with faster pre-computation time exists. Here, we extend our
main constant delay algorithm to count the number of outputs of a deterministic and sequential extended $\VA$ $\cA$ in time
$O(|\cA|\times |d|)$. We also show that counting the number of outputs
of a functional but not necessarily deterministic nor extended $\VA$
is complete for the counting class $\textsc{SpanL}$~\cite{LOGC}, thus
making it unlikely to compute this number efficiently unless
the polynomial hierarchy equals $\textsc{Ptime}$.

\smallskip
\noindent{\bf Related work.}
\input{relatedwork}

\smallskip
\noindent{\bf Organization.} We formally define all the notions used throughout the paper in Section \ref{sec:prelim}. The algorithm for evaluating a deterministic and sequential extended $\VA$ with linear preprocessing and constant delay enumeration is presented in Section \ref{sec:algorithm}, and its application to regular spanners in Section \ref{sec:spanners}. We study the counting problem in Section \ref{sec:lower}, and conclude in Section \ref{sec:concl}. Due to space reasons, most proofs are deferred to the appendix.

%%% Local Variables: 
%%% mode: latex
%%% TeX-master: "../main/main"
%%% End: 

%% file: relatedwork.tex
%!TEX root = ../main/main.tex
%
Constant delay enumeration algorithms (from now on CDAs) for MSO queries have been proposed
in~\cite{bagan2006mso,courcelle2009linear,kazana2013enumeration}. Since
any regular spanner can be encoded by an MSO query (where capture
variables are encoded by pairs of first-order variables), this
implies that CDAs for MSO queries also apply to document
spanners.
%\stijn{What is the preprocessing time in terms of the automaton ?  Giving only the time in the document is not sufficient since we  consider combined complexity precomputation.}
In~\cite{courcelle2009linear}, a CDA was given with preprocessing time
$O(|t| \times \log(|t|))$ in data complexity 
where $|t|$ is the size of the input
structure (e.g. document). In~\cite{kazana2013enumeration}, a CDA
was given based on the deterministic factorization forest
decomposition theorem, a combinatorial result for automata. Our CDA has linear precomputation time over the input document and
does not rely on any previous results, making it
incomparable with \cite{courcelle2009linear,kazana2013enumeration}.

The CDA given by Bagan in~\cite{bagan2006mso} requires a
more detailed comparison. The core algorithm of~\cite{bagan2006mso}
is for a deterministic automaton model which has some resemblance with deterministic $\VA$, but there are several differences.  First of all, Bagan's algorithm is for tree automata and the output are tuples of MSO variables, while our algorithm works only for $\VA$, whose output are first order variables. Second, Bagan's algorithm has preprocessing time $O(|\cA|^3 \times |t|)$, where
$\cA$ is a tree automaton and $t$ is a tree structure.  In contrast,
our algorithm has preprocessing time $O(|\cA| \times |d|)$, namely,
linear in $|\cA|$.  Although Bagan's algorithm is for tree-automata
and this can explain a possible quadratic blow-up in terms of $|\cA|$, it is
not directly clear how to improve its preprocessing time to be linear
in $|\cA|$.  Finally, Bagan's algorithm is described as a composition of high-level operations over automata and
trees, while our algorithm can be described using a few lines of pseudo-code.

There is also recent work~\cite{FKP17,MaturanaRV17} tackling the  enumeration problem
for document spanners directly, but focusing on polynomial delay rather
than constant delay. In \cite{MaturanaRV17}, a complexity theoretic
treatise of polynomial delay (with polynomial pre-processing) is given
for various classes of spanners. And while \cite{MaturanaRV17} focuses
on decision problems that guarantee an existence of a polynomial delay
algorithm, in the present paper we focus on practical algorithms that
furthermore allow for constant delay enumeration. On the other hand,
\cite{FKP17} gives an algorithm for enumerating the results of
a functional $\VA$ automaton $\cA$ over a document $d$ with a delay of
roughly $O(|\cA|^2\times |d|)$, and pre-processing of the order
$O(|\cA|^2 \times |d|)$. The main difference of \cite{FKP17} and the
present paper is that our algorithm can guarantee constant delay,
albeit for a slightly better behaved class of automata. When applying our algorithm
directly to functional $\VA$ as in
\cite{FKP17}, we can still obtain constant delay enumeration, but now
with a pre-processing time of $O(2^{|\cA|}\times |d|)$ (see Section~\ref{sec:spanners}). 
%It is therefore most likely that in the case one wants to consider only functional $\VA$, one would prefer to use the algorithm of \cite{FKP17} when the automaton is large, and when the number of outputs is relatively small, while for spanners that capture a lot of information, or are executed on very big documents, one would be better off using the constant delay algorithm presented here. Another difference is that the algorithm of \cite{FKP17} is presented in terms of automata theoretic constructions, while we aim to give a concise pseudo-code description. %of the algorithm. %, making it easier to implement.
Therefore, if considering only functional $\VA$, the algorithm of \cite{FKP17} would be the preferred option when the automaton is large, and when the number of outputs is relatively small, while for spanners that capture a lot of information, or are executed on very big documents, one would be better off using the constant delay algorithm presented here. Another difference is that the algorithm of \cite{FKP17} is presented in terms of automata theoretic constructions, while we aim to give a concise pseudo-code description.

%%% Local Variables: 
%%% mode: latex
%%% TeX-master: "../main/main"
%%% End: 

%% file: prelim.tex
%!TEX root = ../main/main.tex

%\stijn{TODO: $A \to \cA$}
%\domagoj{PERHAPS ADD AN EXAMPLE THAT FOLLOWS EACH NOTION (I.E. ONE DOCUMENT, DEFINE SPANS OVER IT, AND GIVE A RGX AND VSET THAT EXTRACT INFO FROM IT)}

\noindent{\bf Documents and spans.} We use a fixed finite alphabet $\Sigma$
throughout the paper. A \emph{document}, from which we will extract
information, is a finite string $d = a_1 \dots a_n$ in $\Sigma^*$. We
denote the length $n$ of document $d$ by $|d|$.  % As done in
% previous approaches \cite{FKRV15,AMRV16}, we use the notion of a {\em
%   span} to capture the part of a document $d$ that we wish to
% extract. Formally,
A \emph{span} $s$ is a pair $[i,j\rangle$ of natural
numbers $i$ and $j$ with $1 \leq i \leq j$. Such a span is said to be
\emph{of document $d$} if $j \leq |d|+1$.
% a span $p$
% of a document $d$ is an interval $[i,j\rangle$ such that $1 \leq i
% \leq j \leq |d|+1$, where $|d|$ is the length of the string
% $d$.
In that case, $s$ is associated with a continuous region of the document
$d$ (also called a span of $d$), whose content is the substring of $d$ from positions $i$ to
$j-1$. We denote this substring by $d(s)$ or $d(i,j)$. To illustrate,
Figure~\ref{fig:running-example} shows a document $d$ as well as
several spans of $d$. There, for example, $d(1,5) = \s{John}$.
% Every span $p =
% [i,j\rangle$ of $d$ has an associated content, which is denoted by
% $d(p)$ or $d(i,j)$, and is defined as the substring of $d$ from
% position $i$ to position $j-1$.  
Notice that if $i = j$, then $d(s) = d(i,j) = \varepsilon$. Given two
spans $s_1 = [i_1, j_1\rangle$ and $s_2 = [i_2, j_2\rangle$, if $j_1 =
i_2$ then their concatenation is equal to $[i_1, j_2\rangle$ and is
denoted $s_1 \cdot s_2$.  The set of all spans of
$d$ is denoted by $\sub(d)$.% The set of all spans associated with a document $d$, denoted
% $\sub(d)$, is then defined as the set $\{ [i,j\rangle \mid i,j \in
% \{1, \ldots, |d|+1\} \text{ and } i \leq j\}$.

\medskip
\noindent{\bf Mappings.} Following \cite{MaturanaRV17}, we will use
mappings to model the information extracted from a document. Mappings
differ from tuples (as used by e.g., Fagin et al.~\cite{FKRV15} and
Freydenberger et al.~\cite{F17,freydenberger2016document}) in that not
all variables need to be assigned a span.  Formally, let \(\VV\) be a
fixed infinite set of variables, disjoint from $\Sigma$. A
\emph{mapping} is a function $\mu$ from a finite set of variables
$\dom(\mu) \subseteq \VV$ to spans. 
 % For a document $d$, a \textit{mapping} is a partial function from the set of variables \(\VV\) to \(\sub(d)\). The \textit{domain} \(\dom(\mu)\) of a mapping \(\mu\) is the set of variables for which \(\mu\) is defined. 
Two mappings \(\mu_1\) and \(\mu_2\) are said to be \textit{compatible} (denoted
\(\mu_1 \sim \mu_2\)) if \(\mu_1(x) = \mu_2(x)\) for every \(x\) in
\(\dom(\mu_1) \cap \dom(\mu_2)\). If \(\mu_1 \sim \mu_2\), we define \(\mu_1 \cup
\mu_2\) as the mapping that results from extending \(\mu_1\) with the
values from \(\mu_2\) on all the variables in \(\dom(\mu_2) \setminus
\dom(\mu_1)\). The \textit{empty mapping}, denoted by \(\emptyset\), is the only
mapping such that \(\dom(\emptyset) = \emptyset\). Similarly, \([x \to s]\)
denotes the mapping whose domain only contains the variable \(x\),
which it assigns to be the span~\(s\).
The \textit{join} of two set of mappings \(M_1\) and \(M_2\) is defined as
follows:
\begin{equation*}
  M_1 \Join M_2 = \{\mu_1 \cup \mu_2 \mid \mu_1 \in M_1 \text{, }
  \mu_2 \in M_2 \text{ and } \mu_1 \sim \mu_2\}.
\end{equation*}

%Finally, we say that a mapping \(\mu\) is \textit{hierarchical} if for every
%\(x, y \in \dom(\mu)\), either: \(\mu(x)\) is contained in \(\mu(y)\),
%\(\mu(y)\) is contained in \(\mu(x)\), or \(\mu(x)\) and \(\mu(y)\)
%are disjoint. Similarly, a set of mappings is said to be hierarchical if it only contains hierarchical mappings.

\noindent{\bf Document spanners.} A {\em document spanner} is a
function that maps every input document $d$ to a set of mappings $M$
such that the range of each $\mu \in M$ are spans of $d$---thus
modeling the process of extracting the information (in form of
mappings) from $d$. Fagin et al. \cite{FKRV15} have proposed different
languages for defining spanners: by means of regex formulas, by means
of automata, and by means of  algebra. We next recall the definition
of these languages, and define their semantics in the context of
mappings rather than tuples.

% A number of formaTwo most studied types of
% document spanners are regex formulas and variable-set automata
% \cite{FKRV15}. Since these are the classes of spanners we will
% consider in this paper, we next recall their definition, and their
% semantics in the context of mappings. 

\begin{figure}[b]
\centering\small
{\setlength{\tabcolsep}{0.6mm}
\begin{tabular}{cccccccccccccccccccccccccccccc}
\multicolumn{27}{c}{Document $d$}\\
$\s{J}$ & $\s{o}$ & $\s{h}$ & $\s{n}$ &
$\s{\_}$ & $\s{\langle}$ & $\s{j}$ & $\s{@}$ &
$\s{g}$ & $\s{.}$ & $\s{b}$ &
$\s{e}$ & $\s{\rangle}$ & $\s{,}$ &
$\s{\_}$ & $\s{J}$ & $\s{a}$ & $\s{n}$ &
$\s{e}$ & $\s{\_}$ &
$\s{\langle}$ & $\s{5}$ & $\s{5}$ & $\s{5}$  & $\s{-}$ &
$\s{1}$ & $\s{2}$ & $\s{\rangle}$  \\\hline
\n{1} & \n{2} & \n{3} & \n{4} & 
\n{5} & \n{6} & \n{7} & \n{8} & 
\n{9} & \n{10} & \n{11} & \n{12} & 
\n{13} & \n{14} & \n{15} & \n{16} & 
\n{17} & \n{18} & \n{19} & \n{20} & 
\n{21} & \n{22} & \n{23} & \n{24} & 
\n{25} & \n{26} & \n{27} & \n{28} 
% \\\hline
\end{tabular}}
\begin{tabular}[t]{c|l|l|l}
\multicolumn{4}{c}{$\semd{\gamma}$}\\\hline
% $\str s$-tuple 
& \ccell{name} & \ccell{email} & \ccell{phone}\\\hline
$\mu_1$ & $\mspan{1}{5}$ & $\mspan{7}{13}$ & \\\hline
%$\mu_2$ & $\mspan{1}{5}$ &  & $\mspan{22}{28}$\\\hline
$\mu_2$ & $\mspan{16}{20}$ &  & $\mspan{22}{28}$\\\hline
\end{tabular}
\caption{\label{fig:running-example}A document $d$ and the evaluation $\semd{\gamma}$, with $\gamma$ as defined in Equation~(\ref{eq:example}).}
\end{figure}

\bigskip
\noindent{\bf Regex formulas.} Regex formulas extend the syntax of
classic regular expressions with variable capture expressions of the
form $x\{\gamma\}$. Intuitively, and similar to classical regular
expressions, regex formulas specify a search through an input
document. However, when, during this search, a variable capture
subformula $x\{\gamma\}$ is matched against a substring, the span $s$
that delimits this substring is recorded in a mapping $[x \to s]$ as a
side-effect. Formally, the syntax of \emph{regex formulas} is defined
by the following grammar~\cite{FKRV15}:
$$
  \gamma \coloneqq \varepsilon \mid
    a \mid
    x\{\gamma\} \mid
  \gamma \cdot \gamma \mid
  \gamma \vee \gamma \mid
  \gamma^{*}.
$$
Here, $a$ ranges over letters in $\Sigma$ and $x$ over variables in~$\VV$. We will write $\Var(\gamma)$ to denote  the set of
all variables occurring in regex formula $\gamma$. We write $\RGX$ for
the class of all regex formulas.

% FROM STIJN: It is confusing to introduce variable regexes and regex
% formulas. Let us just call them regex formulas. 
%
% In what follows we will often
% refer to variable regex ($\RGX$, resp.) as a regex formula ($\RGX$
% formula, resp.).

%Intuitively, $\RGX$ use regular expressions to navigate the document, while a subexpression of the form $x\{\gamma\}$ stores a span starting at the current position and matching $\gamma$ into the variable $x$.

The mapping-based spanner semantics of \RGX\ is given in
Table~\ref{tab-semantics} (cf.\@ \cite{MaturanaRV17}).  % The idea
% behind this definition is that we view our expression $\gamma$ as a
% way of defining partial mappings $\mu:\Var(\gamma)\rightharpoonup
% \sub(d)$.
The semantics is defined by structural induction on $\gamma$ and has two
layers. The first layer, $\sempd{\gamma}$, defines the set of all
pairs $(s,\mu)$ with $s \in \sub(d)$ and $\mu$ a mapping such that (1)
$\gamma$ successfully matches the substring $d(s)$ and (2) $\mu$
results as a consequence of this successful match. For example, the
regex formula $\gamma = a$ matches all substrings of input document
$d$ equal to $a$, but results in only the empty mapping. On the other
hand, $\gamma = x\{\gamma_1\}$ matches all substrings that are matched by
$\gamma_1$, but assigns $x$ the span $s$ that delimits the substring being
matched, while preserving the previous variable
assignments. Similarly, in the case of concatenation $\gamma_1\cdot \gamma_2$ we
join the mapping defined on the left with the one defined on the
right, while imposing that the same variable is not used in both parts
(as this would lead to inconsistencies). The second layer,
$\semd{\gamma}$ then simply gives us the mappings that $\gamma$
defines when matching the entire document. Note that when $\gamma$ is
an ordinary regular expression ($\Var(\gamma) = \emptyset)$, then the
empty mapping is output if the entire document matches $\gamma$, and
no mapping is output otherwise. % Furthermore, it is readily verified
% that if $\mu \in \semd{\gamma}$, then $\dom(\mu) \subseteq \Var(\gamma)$.

\begin{table}
\begin{align*}
  \semd{\gamma} &= \{ \mu \mid ((1, |d| + 1), \mu) \in \sempd{\gamma} \}\\
  \sempd{\varepsilon} &= \{ (s, \emptyset) \mid s \in \sub(d) \text{ and } d(s)=\varepsilon\}\\
  \sempd{a} &= \{ (s, \emptyset) \mid s \in \sub(d) \text{ and } d(s) = a \}\\
  \sempd{x\{\gamma\}} &= \{ (s, \mu) \mid \exists (s, \mu') \in \sempd{\gamma}:\\
    &\phantom{{}={}} x \not\in \dom(\mu') \text{ and }
    \mu = [x \to s] \cup \mu' \}\\
  \sempd{\gamma_1 \cdot \gamma_2} &= \{ (s, \mu) \mid
    \exists (s_1, \mu_1) \in \sempd{R_1},\\
    &\phantom{{}={}} \exists (s_2, \mu_2) \in \sempd{\gamma_2}: s = s_1 \cdot s_2, \\   
    &\phantom{{}={}} \dom(\mu_1) \cap \dom(\mu_2) = \emptyset, \text{ and }\\
    &\phantom{{}={}} \mu = \mu_1 \cup \mu_2 \}\\
  \sempd{\gamma_1 \vee \gamma_2} &= \sempd{\gamma_1} \cup \sempd{\gamma_2}\\
  \sempd{\gamma^*} &= \sempd{\varepsilon} \cup \sempd{\gamma} \cup \sempd{\gamma^2} \cup
    \sempd{\gamma^3} \cup \cdots\\
\end{align*}
\vspace*{-25pt}
\caption{The semantics \(\semd{\gamma}\) of a \(\RGX\) \(\gamma\) over a document \(d\). Here \(\gamma^2\) is a shorthand for \(\gamma \cdot \gamma\), similarly \(\gamma^3\) for \(\gamma \cdot \gamma \cdot \gamma\), etc.}
\label{tab-semantics}
\end{table}

\begin{example}
  \label{ex:regex}
  
  Consider the task of extracting names, email addresses and phone numbers from documents. To do this
  we could use the regex formula $\gamma$ defined as
  \begin{equation}\label{eq:example}
    \Sigma^* \cdot \textit{name}\{\gname\} \cdot \_ \cdot \langle\cdot(\textit{email}\{\gemail\} \vee \textit{phone}\{ \gphone)\})\cdot \rangle \cdot \Sigma^*
  \end{equation}
  where $\s{\_}$ represents a space; $name$, $email$, and $phone$ are variables; and $\gname$, $\gemail$, and $\gphone$ are regex formulas that recognize person names, email addresses, and phone numbers, respectively. We omit the particular definition of these formulas as this is irrelevant for our purpose.
  % More precisely,
  % \begin{align*}
  %   \gname & := \s{[A\text{-}Z]}\cdot\s{[a\text{-}z]}^*\\
  %   \gemail & := \s{[a\text{-}z]}^+ \cdot \s{@} \cdot
  %   \s{[a\text{-}z]}^+ \cdot \s{.} \cdot \s{[a\text{-}z]}^+\\
  %   \gphone & := \s{[0\text{-}9]}^+ \cdot \s{\text{-}} \cdot \s{[0\text{-}9]}^+.
  % \end{align*}
  % Here, shorthand notations like $\s{[A\text{-}Z]}$ denote the
  % disjunction $\s{A}\lor\dots\lor\s{Z}$ and $\gamma^+$ abbreviates
  % $\gamma \cdot \gamma^*$. 
  % The symbol $\s{\_}$ represents a space. 
  The result $\semd{\gamma}$ of evaluating
  $\gamma$ over the document~$d$ shown in
  Figure~\ref{fig:running-example} is shown at the bottom of
  Figure~\ref{fig:running-example}.
\end{example}

It is worth noting that the syntax of regex formula used here is
slightly more liberal than that used by Fagin et al.~\cite{FKRV15}. In
particular Fagin et al.\@ require regex formulas to adhere to certain
syntactic restrictions that ensure that the formula is
\emph{functional}: every mapping in $\semd{\gamma}$ is defined on all
variables appearing in $\gamma$, for every $d$. For regex formulas
that satisfy this syntactic restriction, the semantics given here
coincides with that of Fagin et al~\cite{FKRV15} (see
\cite{MaturanaRV17} for a detailed discussion).

\medskip
\noindent{\bf Variable-set automata.} A \emph{variable-set automaton}
(\(\VA\)) \cite{FKRV15} is an finite-state automaton extended with captures variables in a
way analogous to \RGX; that is, it behaves as a usual finite state
automaton, except that it can also open and close variables.
Formally, a $\VA$ automaton $\cA$ is a tuple \((Q, q_0, F, \delta)\),
where \(Q\) is a finite set of \textit{states}; $q_0 \in Q$ is the
initial state; $F\subseteq Q$ is the set of final states; and $\delta$
is a \textit{transition relation} consisting of \emph{letter
  transitions} of the form $(q, a, q')$ and \emph{variable
  transitions} of the form $(q, \Open{x}, q')$ or $(q, \Close{x},q')$,
where $q, q' \in Q$, $a \in \Sigma$ and $x \in \VV$.  The \(\vdash\)
and \(\dashv\) are special symbols to denote the opening or closing of
a variable $x$. We refer to $\Open{x}$ and $\Close{x}$ collectively as
\emph{variable markers}.  We define the set $\Var(\cA)$ as the set of
all variables $x$ that are mentioned in some transition of $\cA$.

A configuration of a $\VA$ automaton over a document $d$ is a tuple $(q, i)$ where $q \in Q$ is the current state and \(i \in [1, |d| + 1]\) is the \textit{current position} in $d$.
A run $\rho$ over a document $d = a_1 a_2 \cdots a_n$ is a sequence of the form:
$$
\rho \ = \ (q_0, i_0) \ \trans{o_1} \ (q_1, i_1) \ \trans{o_2} \ \cdots \ \trans{o_m} \ (q_m, i_m)
$$
where $o_j \in \Sigma \cup \{\Open{x}, \Close{x} \mid x \in \VV\}$ and
$(q_j, o_{j+1}, q_{j+1}) \in \delta$. Moreover, $i_0, \ldots, i_n$ is a non-decreasing sequence such that $i_0 = 1$, $i_m = |d|+1$, and $i_{j+1} =
i_j +1$ if $o_{j+1} \in \Sigma$ (i.e.\ the automata moves one position
in the document only when reading a letter) and $i_{j+1} = i_j$
otherwise. Furthermore, we say that a run $\rho$ is \emph{accepting}
if $q_m \in F$ and that it is \emph{valid} if variables are opened and closed
in a correct manner (that is, each $x$ is opened or closed at most
once, and $x$ is opened at some position $i$ if and only if it is closed at
some position $j$ with $i \leq j$).
Note that not every accepting run is valid. In case that
$\rho$ is both accepting and valid, we define $\mu^{\rho}$ to be the
mapping that maps $x$ to $[i_j, i_k\rangle \in \sub(d)$ if, and only
if, $o_{i_j} = \Open{x}$ and $o_{i_k} = \Close{x}$ in~$\rho$.
Finally, the semantics of $\cA$ over $d$, denoted by \(\semd{\cA}\) is
defined as the set of all $\mu^{\rho}$ where $\rho$ is a valid and
accepting run of $\cA$ over $d$.

Note that validity requires only that variables are
opened and closed in a correct manner; it does not require that all
variables in $\Var(\cA)$ actually appear in the run.  Valid runs that do
mention all variables in $\Var(\cA)$ are called \emph{functional}. In a
functional run, all variables are hence opened and closed exactly once
(and in the correct manner) whereas in a valid run they are opened and
closed at most once. % As such, $\dom(\mu^\rho) = \var(\cA)$ for every
% functional accepting run $\rho$, whereas we only know $\dom(\mu^\rho)
% \subseteq \var(\cA)$ for valid accepting $\rho$.

A VA $\cA$ is \emph{sequential} (\sVA) if every accepting
run of $\cA$ is valid. It is \emph{functional} (\fVA) if every accepting run is
functional. In particular, every $\fVA$ is also
sequential. Intuitively, during a run a $\sVA$ does not need to check
whether variables are opened and closed in a correct manner; the run is guaranteed to be valid whenever a final state is
reached.

% \begin{definition}
%   VA $A$ is \emph{sequential} if every accepting run of $A$ is valid.
%   It is \emph{functional} if it is sequential and every variable $x
%   \in \Var(A)$ is opened in every accepting run of $A$.
% \end{definition}

%\stijn{Cristian: can you add to the following paragraph the
%  preprocessing time required ?}

It was shown in \cite{MaturanaRV17,FKP17} that constant delay enumeration
(after polynomial-time preprocessing) is not possible for variable-set automata in general. However, the authors in \cite{MaturanaRV17} also show that for the class of $\fVA$ or $\sVA$, polynomial delay enumeration is possible, thus leaving open the question of constant delay in this case. 
As we will see, the sequential property is important in order to have constant-delay algorithms. % in general.
%To define the class of {\em functional $\VA$} \cite{F17,MaturanaRV17}, we first need to define the notion of a path. A path $\pi$ of a variable-set automaton $A = (Q, q_0, F, \delta)$ is a finite sequence of transitions $\pi: (q_1, s_2, q_2), (q_2, s_3, q_3) \ldots, (q_{m-1}, s_m, q_m)$ such that $(q_i, s_{i+1}, q_{i+1}) \in
%\delta$ for all $i \in [1, m - 1]$.
%We say that a path $\pi$ of $A$ is functional if for every variable \(x \in
%\VV\) it holds that there is at most one \(i \in [1, m]\) such that \(s_i = \Open{x}\); and if such an $i$ exists, then there is precisely one $j\geq i$ such that \(s_j = \Close{x}\). We say that variable-set automaton $A$ is functional if every path from the initial to a final state in $A$ is functional\footnote{Strictly speaking, \cite{MaturanaRV17} defines a slightly more general class of automata called sequential $\VA$, that are not required to close an open variable on a path. Since the functionality constraint seems more natural, and is better established in the literature, in this paper we will consider functional $\VA$ instead of sequential.}.

\medskip
\noindent{\bf Spanner algebras.} In addition to defining basic
document spanners through $\RGX$ or $\VA$, practical information
extraction systems also allow spanners to be defined by applying basic
algebraic operators on already existing spanners. This is formalized
as follows. Let $\mathcal{L}$ be a language for defining document
spanners (such as $\RGX$ or $\VA$). Then we denote by
$\mathcal{L}^{\{\pi,\cup,\Join\}}$ the set of all expressions
generated by the following grammar:
\[ e := \alpha \mid \pi_Y(e) \mid e \cup e \mid e \Join e. \] Here,
$\alpha$ ranges over expressions of $\mathcal{L}$, and $Y$ is a finite subset of $\VV$. Assume that
$\sem{\alpha}$ denotes the spanner defined by $\alpha \in
\mathcal{L}$. Then the semantics $\sem{e}$ of expression $e$ is the
spanner inductively defined as follows:
\begin{align*}
  \semd{\pi_Y(e)} & = \{\mu|_Y : \mu \in \semd{e}\} \\
  \semd{e_1 \cup e_2} & = \semd{e_1} \cup \semd{e_2} \\
  \semd{e_1 \Join e_2} & = \semd{e_1} \Join \semd{e_2}
\end{align*}
Here, $\mu|_Y$ is the restriction of $\mu$ to variables $Y$
and $\semd{e_1} \Join \semd{e_2}$ is the join of two sets of mappings.  

It was shown by Fagin et al.\@~\cite{FKRV15} that $\VA$,
$\RGX^\algops$, and $\VA^\algops$ all express the same class of
spanners, called \emph{Regular
  Spanners}. In particular, every expression in $\RGX^\algops$, and
$\VA^\algops$ is equivalent to a VA. This will be used later
in Section~\ref{sec:spanners}.

\medskip
\noindent{\bf The enumeration problem.} In this paper, we study the
problem of enumerating all mappings in $\semd{\gamma}$, given a document spanner $\gamma$ (e.g. by means of a $\VA$) and
a document $d$. Given a language
$\LL$ for document spanners we define the main enumeration problem of
evaluating expressions from $\LL$ formally as follows:
\begin{center}
	\framebox{
		\begin{tabular}{rl}
			\textbf{Problem:} & $\ENUM[\LL]$\\
			\textbf{Input:} &  Expression $\gamma \in
                        \LL$ and  document $d$.  \\
			\textbf{Output:} & All mappings in
                        $\semd{\gamma}$ without \\ &repetitions. \\
		\end{tabular}
	}
\end{center}
As usual, we assume that the size $|R|$ of a $\RGX$ expression $R$ is
the number of alphabet symbols and operations, and the size $|\cA|$ of a VA $\cA$ is given by the number of transitions plus the number of states. 
Furthermore, the size $|e|$ of an expression $e$ in $\mathcal{L}^{\{\pi,\cup,\Join\}}$ (e.g $\RGX^\algops$) is given by $\sum_i |\alpha_i|$ where $\alpha_i$ are the expressions in $\mathcal{L}$ plus the number of operators (i.e. $\algops$) used in $e$.

\medskip
\noindent{\bf Enumeration with constant delay.} 
We use the definition of constant delay enumeration presented in~\cite{Segoufin13,segoufin2014glimpse,Segoufin15} adapted to $\ENUM[\LL]$. 
As it is standard in the literature~\cite{Segoufin13}, we consider enumeration algorithms over Random Access Machines (RAM) with addition and uniform cost measure~\cite{AhoHU74}. 
Given a language $\LL$ for document spanners, we say that an enumeration algorithm~$\mathcal{E}$ for $\ENUM[\LL]$ has constant delay if $\mathcal{E}$ runs in two phases over the input $\gamma \in \LL$ and~$d$.
\begin{compactitem}[-]
	\item The first phase (\emph{precomputation}) which does not produce output. 
	\item The second phase (\emph{enumeration}) which occurs
          immediately after the precomputation phase and enumerates
          all mappings in $\semd{\gamma}$ without repetitions. We
          require that the delay between the start of enumeration,
          between any two consecutive outputs, and between the last
          output and the end of this phase depend only on
          $|\gamma|$. A such, it is constant in $|d|$.
\end{compactitem}
We say that $\mathcal{E}$ is a \emph{constant delay algorithm} for $\ENUM[\LL]$ with precomputation phase $f(|\gamma|, |d|)$, if $\mathcal{E}$ has constant delay and the precomputation phase takes time $O(f(|\gamma|, |d|))$. We say that $\mathcal{E}$ features constant delay enumeration after linear time pre-processing if $f(|\gamma|, |d|)=g(|\gamma|)\cdot|d|$ for some function $g$. It is important to stress that the delay between consecutive outputs has to be constant, so we seek to reduce the precomputation time $f(|\gamma|, |d|)$ as much as possible.

% to  taking constant time between two consecutive writes in the output tape, including constant  time until the first register is written.
%Note that, in contrast to the definition given
%in~\cite{segoufin2014glimpse}, we restrict to have constant time
%between two write operations rather than constant time between two
%ouputs (e.g. mappings). 

%
%\cristian{PREVIOUS VERSION}
%In this paper, we use a slightly more strict definition of constant delay algorithm than the one presented in~\cite{segoufin2014glimpse}. 
%As it is standard in the literature, we consider enumeration algorithms over Random Access Machines (RAM) with addition and uniform cost measure~\cite{AhoHU74}, plus a write-once output tape (i.e. every register can be written at most once in the output). Then we say that an enumeration algorithm on a RAM model has constant delay if it consists of two phases. The first phase (called \emph{precomputation} phase) takes polynomial time in the size of the input and, further, nothing is output during this process. 
%The second step (called enumeration phase) enumerates the whole output using what was precomputed in the previous phase taking constant time between two consecutive writes in the output tape, including constant  time until the first register is written.
%Note that, in contrast to the definition given
%in~\cite{segoufin2014glimpse}, we restrict to have constant time
%between two write operations rather than constant time between two
%ouputs (e.g. mappings). 

%%% Local Variables: 
%%% mode: latex
%%% TeX-master: "../main/main"
%%% End: 

%% file: cde-algorithm.tex
%!TEX root = ../main/main.tex

In this section we present an algorithm featuring constant delay enumeration after linear pre-processing for a syntactic variant of $\VA$ that we call extended variable-set automata ($\EVA$ for short). This variant avoids several problems that $\VA$ have in terms of evaluation. Later, in Section~\ref{sec:spanners}, we show how this algorithm can be applied to ordinary $\VA$, $\RGX$ formulas, and spanner algebras. We start by introducing extended $\VA$.

\subsection{Extended variable-set automata}\label{ss:extended}

$\VA$ can open or close variables in arbitrary ways, which can lead to multiple runs that define the same output. An example of this is given in Figure \ref{fig:badauto}, where we have a functional $\VA$ (\fVA) that has two runs resulting in the same output (i.e. they produce a mapping that assigns the entire document both to $x$ and $y$). This is of course problematic for constant delay enumeration, as outputs must be enumerated without repetitions\footnote{As shown in \cite{FKRV15}, such behaviour also leads to a factorial blow-up when defining the join of two $\VA$, as all possible orders between variables need to be considered. See Section~\ref{sec:spanners} for further discussion.}. 

Ideally, when running a $\VA$ one would like to start by declaring which variable operations take place before reading the first letter of the input word, then process the letter itself, followed by another step declaring which variable operations take place after this, read the next letter, etc. Extended variable-set automata achieve this by allowing multiple variable operations to take place during a single transition, and by forcing each transition that manipulates variables to be followed by a transition processing a letter from the input word.

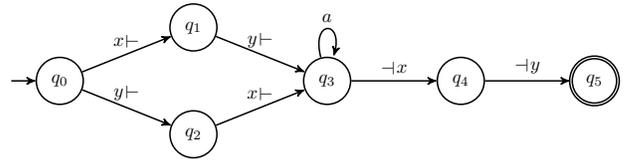
\begin{figure}
\begin{center}
 \resizebox{\columnwidth}{!}{
		\begin{tikzpicture}[->,>=stealth',auto, thick, scale = 1.0,initial text= {}]
 
		  % The graph
		  \node [state,initial] at (0,0) (q0) {$q_0$};
		  \node [state] at (2.5,1) (q1) {$q_1$};
		  \node [state] at (2.5,-1) (q2) {$q_2$};		  
		  \node [state] at (5,0) (q3) {$q_3$};
		  \node [state] at (7.5,0) (q4) {$q_4$};
		  \node [state,accepting] at (10,0) (q5) {$q_5$};
  
		  % Graph edges
		  \path[->] (q0) edge node[above] {$\Open{x}$} (q1);
		  \path[->] (q0) edge node[above] {$\Open{y}$} (q2);
		  \path[->] (q1) edge node[above] {$\Open{y}$} (q3);
		  \path[->] (q2) edge node[above] {$\Open{x}$} (q3);		  
		  \path[->] (q3) edge[loop above] node[above] {$a$} (q3);
		  \path[->] (q3) edge node[above] {$\Close{x}$} (q4);		  
		  \path[->] (q4) edge node[above] {$\Close{y}$} (q5);		  
		 \end{tikzpicture}
	}
\end{center}
\vspace*{-15pt}
\caption{A functional $\VA$ with multiple runs defining the same output mapping.}
\label{fig:badauto}
\end{figure}

Formally, let $\markers_{\VV} = \{\Open{x}, \Close{x} \mid x \in \VV \}$ be the set of open and close markers for all the variables in~$\VV$. An {\em extended variable-set automaton} (extended $\VA$, or $\EVA$) is a tuple $\cA = (Q, q_0, F, \delta)$, where $Q$, $q_0$, and $F$ are the same as for variable-set automata, and $\delta$ is the transition relation consisting of letter transitions $(q,a,q')$, or \emph{extended variable transitions} $(q, S, q')$, where $S \subseteq \markers_{\VV}$ and $S \neq \emptyset$. 
A run $\rho$ over a document $d = a_1 a_2 \cdots a_n$ is a sequence of the form:
\begin{equation}\label{eq:run}
  \rho \ = \ q_0 \ \trans{S_1} \ p_0 \ \trans{a_1} \ q_1 \ \trans{S_2} \ p_1 \ \trans{a_2} \ \ldots \ \ \trans{a_n} \ q_n \ \trans{S_{n+1}} p_n
\end{equation}
where every $S_i$ is a (possibly empty) set of markers, $(p_i, a_{i+1},
q_{i+1}) \in \delta$, and $(q_i, S_{i+1}, p_i) \in \delta$ whenever
$S_{i+1} \neq \emptyset$, and $q_i = p_i$ otherwise. Notice that
extended variable transitions and letter transitions must alternate in
a run of an \EVA, and that a transition with the $\emptyset$
of variable markers is only allowed when it stays in the same
position.

As in the setting of ordinary \VA, we say that a run $\rho$ is
\emph{valid} if variables are opened and closed in a correct manner:
the sets $S_i$ are pairwise disjoint; for every $i$ and every
$\Open{x}\, \in S_i$ there exists $j \geq i$ with $\Close{x}\, \in
S_j$; and, conversely, for every $j$ and every $\Close{x}\, \in S_j$
there exists $i \leq j$ with $\Open{x}\, \in S_i$.  For a valid run
$\rho$ we define the mapping $\mu^{\rho}$ that maps $x$ to $[i,
j\rangle \in \sub(d)$ if, and only if, $\Open{x}\, \in S_i$,
$\Close{x}\, \in S_j$ and $i \leq j$.  Also, we say that $\rho$ is
\emph{accepting} if $p_n \in F$.  Finally, the semantics of $\cA$ over
$d$, denoted by \(\semd{\cA}\) is defined as the set of all mappings
$\mu^{\rho}$ where $\rho$ is a valid and accepting run of $\cA$ over
$D$. We transfer the notion of being \emph{sequential} (\sEVA) and
\emph{functional} (\fEVA) from normal \VA\ to extended \VA\ in the obvious way.

An extended $\VA$ $\cA$ is \emph{deterministic} if the transition relation
$\delta$ of $A$ is a partial function $\delta: Q \times (\Sigma
\cup 2^{\markers_{\VV}} \backslash \{\emptyset\}) \rightarrow Q$. If
$\cA$ is deterministic, then we define $\markers_\delta(q)$ as the set $\{S
\subseteq \markers_{\VV} \mid (q, S) \in \dom(\delta)\}$. Note that,
in contrast to determinism for classical NFAs, determinism as defined
here does not imply that there is at most one run for each input
document~$d$. Instead, it implies that for every document $d$ and
every $\mu \in \semd{\cA}$, there is exactly one valid and accepting run $\rho$ with $\mu
= \mu^\rho$. In other words: there may still be many valid accepting runs on
a document $d$, but each such run defines a unique mapping.
%\cristian{Add an example here of an extended VA. }
For instance, we could convert the $\VA$ $\cA$ from Figure \ref{fig:badauto} into an equivalent $\EVA$ $\cA'$ by adding a transition $(q_0,\{\Open{x},\Open{y}\},q_3)$ to $\delta$, and removing the states $q_1$ and $q_2$, together with their associated transitions. It is easy to see that $\cA'$ is deterministic, so all accepting runs will define an unique mapping, thus avoiding the issues that $\cA$ has when considering the enumeration of output mappings.

The following results shows that $\EVA$ are indeed a natural variant of normal
\VA\ and that all $\EVA$ can be determinized.
\begin{theorem}\label{theo:equivalence}
For every \VA\ $\cA$ there exists an \EVA\ $\cA'$ such that $\cA \equiv \cA'$ and vice versa. Furthermore, if $\cA$ is sequential (resp. functional), then $\cA'$ is also sequential (resp. functional).
\end{theorem}
\begin{proposition}\label{prop:determinization}
For every \EVA\ $\cA$ there exists a deterministic 
\EVA\ $\cA'$  such that $\cA \equiv \cA'$. %Furthermore, the size of $\cA'$ is at most exponential in the size of $\cA$.
\end{proposition}

In Section~\ref{sec:spanners} we will study in detail the complexity of these translations; to present our algorithm we only require equivalence between the models.
 % \stijn{Should we not add here that more insight into the complexity
 %   of both conversions will be given in Section 4? At current, the
 %   reader will be wondering about this.}

%\cristian{Here we should be more precise with the size of the automaton with respect to the transitions (both in Theorem~\ref{theo:equivalence} and Proposition~\ref{prop:determinization}). This is also used later in Section 4. Maybe, we can state both results as it is now, and refine the size of the automaton later in Section 4.}

\subsection{Constant delay evaluation algorithm}\label{ss:cda}
%\begin{itemize}
%\item start with Stijn's intuitive description
%\item The algorithm
%\item Example
%\item Sketch of the correctness proof (explain why each step does what it does, and why is it correct and can be implemented as we claim)
%\item Counting the number of solutions is in PTIME
%\end{itemize}

The objective of this section is to describe an algorithm that takes as input
a {\em deterministic and sequential} $\EVA$ $\cA$ (deterministic $\sEVA$ for short)
and a document $d$, and enumerates the set $\semd{\cA}$ with a constant
delay after pre-processing time $O(|\cA|\times |d|)$.  We start with an
intuitive explanation of the algorithm's underlying idea, and then
give the full algorithm.

\subsubsection{Intuition} 
\label{sec:intuition}
%\medskip
%\noindent{\bf Intuition.}
As with the majority of constant delay algorithms, in the
pre-processing step we build a compact representation of the output
that is used later in the enumeration step. In our case, we build a
directed acyclic graph (\dag) that can then be traversed in a
depth-first manner to enumerate all the output mappings. This \dag
will encode all the runs of $\cA$ over $d$, and its construction can
be summarized as follows:
\begin{itemize}\itemsep=0pt
\item Convert the input word $d$ into a deterministic extended $\VA$ $\cA_d$;
\item Build the product between $\cA$ and $\cA_d$, and annotate the variable transitions with the position of $d$ where they take place;
\item Replace all the letters in the transitions of $\cA\times \cA_d$ with $\varepsilon$, and construct the ``forward" $\varepsilon$-closure of the resulting graph.
\end{itemize}

\begin{figure}[t]
\begin{center}
 \resizebox{\linewidth}{!}{
		\begin{tikzpicture}[->,>=stealth',auto, thick, scale = 1.0, initial text= {}]
 
		  % The graph
		  \node [state,initial] at (0,0) (q0) {$q_0$};
		  \node [state] at (2,1) (q1) {$q_1$};
		  \node [state] at (4,1) (q4) {$q_4$};
		  \node [state] at (6,1) (q6) {$q_6$};
		  \node [state] at (2,-1) (q2) {$q_2$};
		  \node [state] at (4,-1) (q5) {$q_5$};
		  \node [state] at (6,-1) (q7) {$q_7$};
		  \node [state] at (8,0) (q8) {$q_8$};
		  \node [state] at (5,-3) (q3) {$q_3$};
		  \node [state,accepting] at (10,0) (q9) {$q_9$};

		  % Graph edges
		  \path[->]
		  (q3) edge[loop above] node[above] {$a,b$}   (q3);
		  \path[->] (q0) edge node[above] {$\Open{x}$} (q1);
		  \path[->] (q0) edge node[above] {$\Open{y}$} (q2);
		  \path[->] (q0) edge[bend right] node[right = 5pt] {$\Open{x},\Open{y}$} (q3);
		  \path[->] (q3) edge[bend right] node[left = 5pt] {$\Close{x},\Close{y}$} (q9);
		  \path[->] (q1) edge node[above] {$a$} (q4);
		  \path[->] (q2) edge node[above] {$a$} (q5);
		  \path[->] (q6) edge node[above] {$b$} (q8);
		  \path[->] (q7) edge node[above] {$b$} (q8);
		  \path[->] (q4) edge node[above] {$\Open{y}$} (q6);
		  \path[->] (q1) edge node[above] {$a$} (q4);
		  \path[->] (q5) edge node[above] {$\Open{x}$} (q7);
		  \path[->] (q8) edge node[above=0.2] {$\Close{x},\Close{y}$} (q9);
		 \end{tikzpicture}
	}
\end{center}
\caption{An extended functional $\VA$ $\cA$.}
\label{fig-automaton}
\end{figure}

We first illustrate
how this construction works by means of an example. For this, consider
the $\EVA$ $\cA$ from Figure \ref{fig-automaton}. It
is straightforward to check that this automaton is functional (hence sequential) and
deterministic. To evaluate $\cA$ over document $d=ab$ we first convert the input document $d$ into an $\EVA$ $\cA_d$ that represents all possible ways of assigning spans
over $d$ to the variables of $\cA$. The automaton $\cA_d$ is a chain
of $|d|+1$ states linked by the transitions that spell out the word $d$. That is, $\cA_d$ has the states $p_1,\ldots ,p_{|d|+1}$, and
letter transitions $(p_i,d_{i},p_{i+1})$, with $i=1\ldots |d|$, and
where $d_i$ is the $i$th symbol of $d$. 
Furthermore, each state $p_i$ has $2^{|\Var(\cA)|-1}$ self loops, each
labelled by a different non-empty subset of $\markers_{\Var(A)}$. For
instance, in the case of $d=ab$, the automaton $\cA_d$ is the following:

\begin{center}
\resizebox{0.6\linewidth}{!}{
		\begin{tikzpicture}[->,>=stealth',auto, thick, scale = 1.0, initial text= {}]
 
		  % The graph
		  \node [state,initial] at (0,0) (p0) {$p_1$};
		  \node [state] at (2,0) (p1) {$p_2$};
		  \node [state, accepting] at (4,0) (p2) {$p_3$};
  
		  % Graph edges
		  \path[->] (p0) edge node[above] {$a$} (p1);
		  \path[->] (p1) edge node[above] {$b$} (p2);

		  \path[->] (p0) edge[out=160, in=100, distance=1.5cm] node[above] {$\{\Open{x}\}$}   (p0);

		  \path[->] (p0) edge[out=20, in=80, distance=1.5cm] node[above] {$\{\Open{y}\}$}   (p0);

		  \path[->] (p0) edge[out=240, in=190, distance=1.5cm]
                  node[below] {$\{\Close{x}, \Open{y}\}$}   (p0);

		  \path[->] (p0) edge[out=340, in=280, distance=1.5cm]
                  node[right] {$\dots$}   (p0);

                  \path[->] (p1) edge[out=120, in=60, distance=1.5cm] node[above] {$\dots$}   (p1);

                  \path[->] (p2) edge[out=120, in=60, distance=1.5cm] node[above] {$\dots$}   (p2);
		 \end{tikzpicture}
}
\end{center}
\vspace{-2ex}
Next, we build the product automaton $\cA \times \cA_d$ in the
standard way (i.e. by treating variable transitions as letters and
applying the NFA product construction). During construction, we take
care to only create product states of the form $(q, p)$ that are
reachable from the initial product state $(q_0,p_1)$. In addition, we annotate
the variable transitions of this automaton with the position in $d$
where the particular transition is applied. For this, we use the fact
that $\cA_d$ is a chain of states, so in the product $\cA\times
\cA_d$, each variable transition is of the form
$((q,p_i),S,(q',p_i))$. We therefore annotate the set $S$ with the
number $i$. We depict the resulting annotated product automaton for
$\cA$ and $d=ab$ in Figure~\ref{fig:numbered} (top).

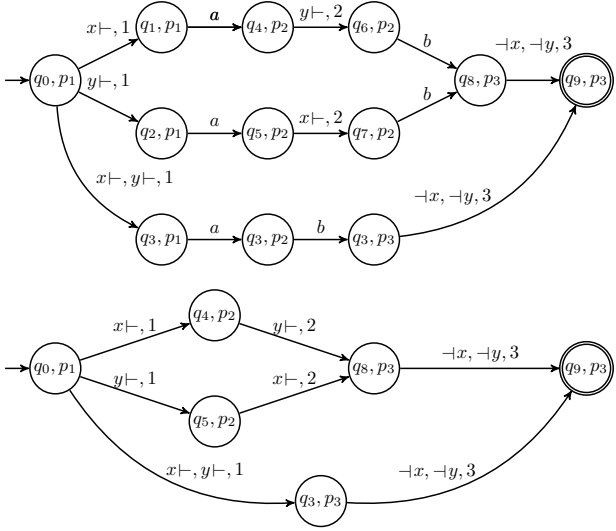
\begin{figure}[t]
\begin{center}
 \resizebox{\linewidth}{!}{
		\begin{tikzpicture}[->,>=stealth',auto, thick, scale = 1.0, initial text={},every state/.style={inner sep=0.5mm}]
 
		  % The graph
		  \node [state,initial] at (0,0) (q0) {$q_0, p_1$};
		  \node [state] at (2,1) (q1) {$q_1, p_1$};
		  \node [state] at (4,1) (q4) {$q_4, p_2$};
		  \node [state] at (6,1) (q6) {$q_6, p_2$};
		  \node [state] at (2,-1) (q2) {$q_2, p_1$};
		  \node [state] at (4,-1) (q5) {$q_5, p_2$};
		  \node [state] at (6,-1) (q7) {$q_7, p_2$};
		  \node [state] at (8,0) (q8) {$q_8, p_3$};
		  \node [state] at (2,-3) (q3) {$q_3, p_1$};
		  \node [state] at (4,-3) (q3p1) {$q_3, p_2$};
		  \node [state] at (6,-3) (q3p2) {$q_3, p_3$};
		  \node [state,accepting] at (10,0) (q9) {$q_9,p_3$};
  
		  % Graph edges
		  \path[->] (q0) edge node[above=0.2] {$\Open{x}, 1$} (q1);
		  \path[->] (q0) edge node[above=0.2] {$\Open{y}, 1$} (q2);
		  \path[->] (q0) edge[bend right] node[right = 5pt] {$\Open{x},\Open{y},1$} (q3);
                  \path[->] (q3) edge node[above] {$a$} (q3p1);
                  \path[->] (q3p1) edge node[above] {$b$} (q3p2);
		  \path[->] (q3p2) edge[bend right] node[left = 5pt]
                  {$\Close{x},\Close{y},3$} (q9);
		  \path[->] (q1) edge node[above] {$a$} (q4);
		  \path[->] (q2) edge node[above] {$a$} (q5);
		  \path[->] (q6) edge node[above] {$b$} (q8);
		  \path[->] (q7) edge node[above] {$b$} (q8);
		  \path[->] (q4) edge node[above] {$\Open{y},2$} (q6);
		  \path[->] (q1) edge node[above] {$a$} (q4);
		  \path[->] (q5) edge node[above] {$\Open{x}, 2$} (q7);
		  \path[->] (q8) edge node[above=0.4]
                  {$\Close{x},\Close{y}, 3$} (q9);
		 \end{tikzpicture}
	}
\end{center}
\vspace{-5ex}
\begin{center}
 \resizebox{\linewidth}{!}{
		\begin{tikzpicture}[->,>=stealth',auto, thick, scale = 1.0, initial text= {},every state/.style={inner sep=0.5mm}]
 
		  % The graph
		  \node [state,initial] at (0,0) (q0) {$q_0, p_1$};
		  \node [state] at (3,1) (q1) {$q_4, p_2$};

		  \node [state] at (6,0) (q6) {$q_8, p_3$};
		  \node [state] at (3,-1) (q2) {$q_5, p_2$};

		  \node [state] at (5,-2.5) (q3) {$q_3, p_3$};
		  \node [state,accepting] at (10,0) (q9) {$q_9,p_3$};
  
		  % Graph edges
		  \path[->] (q0) edge node[above] {$\Open{x}, 1$} (q1);
		  \path[->] (q0) edge node[above] {$\Open{y}, 1$} (q2);
		  \path[->] (q1) edge node[above] {$\Open{y}, 2$} (q6);
		  \path[->] (q2) edge node[above] {$\Open{x}, 2$} (q6);		  
		  \path[->] (q6) edge node[above] {$\Close{x},\Close{y}, 3$} (q9);		  		  
		  \path[->] (q0) edge[bend right] node[above right = -5pt] {$\Open{x},\Open{y},1$} (q3);
		  \path[->] (q3) edge[bend right] node[above left = -5pt] {$\Close{x},\Close{y},3$} (q9);
%		  \path[->] (q1) edge node[above] {$\Open{y},1$} (q6);
%		  \path[->] (q2) edge node[above] {$\Open{x}, 1$} (q7);
%
%                  \path[->] (q6) edge node[above]
%                  {$\Close{x},\Close{y}, 2$} (q9);
%                  \path[->] (q7) edge node[above]
%                  {$\Close{x},\Close{y}, 2$} (q9);
		 \end{tikzpicture}
	}
\end{center}
\vspace{-2ex}
\caption{The annotated product automaton (top) and its
  ``forward''$\varepsilon$-closure (bottom).}% Solid lines represent the transitions before applying the $\varepsilon$-closure, and dashed lines the transitions added during the $\varepsilon$-closure.}
\label{fig:numbered}
\end{figure}

% \begin{figure}[t]
% \begin{center}
%  \resizebox{\linewidth}{!}{
% 		\begin{tikzpicture}[->,>=stealth',auto, thick, scale = 1.0, initial text= {},every state/.style={inner sep=0.5mm}]
 
% 		  % The graph
% 		  \node [state,initial] at (0,0) (q0) {$q_0, p_1$};
% 		  \node [state] at (3,1) (q1) {$q_4, p_2$};

% 		  \node [state] at (6,0) (q6) {$q_8, p_3$};
% 		  \node [state] at (3,-1) (q2) {$q_5, p_2$};

% 		  \node [state] at (5,-2.5) (q3) {$q_3, p_3$};
% 		  \node [state,accepting] at (10,0) (q9) {$q_9,p_3$};
  
% 		  % Graph edges
% 		  \path[->] (q0) edge node[above] {$\Open{x}, 1$} (q1);
% 		  \path[->] (q0) edge node[above] {$\Open{y}, 1$} (q2);
% 		  \path[->] (q1) edge node[above] {$\Open{y}, 2$} (q6);
% 		  \path[->] (q2) edge node[above] {$\Open{x}, 2$} (q6);		  
% 		  \path[->] (q6) edge node[above] {$\Close{x},\Close{y}, 3$} (q9);		  		  
% 		  \path[->] (q0) edge[bend right] node[above right = -5pt] {$\Open{x},\Open{y},1$} (q3);
% 		  \path[->] (q3) edge[bend right] node[above left = -5pt] {$\Close{x},\Close{y},3$} (q9);
% %		  \path[->] (q1) edge node[above] {$\Open{y},1$} (q6);
% %		  \path[->] (q2) edge node[above] {$\Open{x}, 1$} (q7);
% %
% %                  \path[->] (q6) edge node[above]
% %                  {$\Close{x},\Close{y}, 2$} (q9);
% %                  \path[->] (q7) edge node[above]
% %                  {$\Close{x},\Close{y}, 2$} (q9);
% 		 \end{tikzpicture}
% 	}
% \end{center}
% \caption{The ``forward'' $\varepsilon$-closure of $\cA\times \cA_d$. }
% \label{fig:epsilon-closure}
% \end{figure}
In the final step, we replace all letter transitions with
$\varepsilon$-transitions and compute what we call the ``forward''
$\varepsilon$-closure. This is done by considering each variable
transition $((q,p),(S,i),(q',p'))$ of the annotated product automaton, and then computing all the states $(r,s)$ such that
one can reach $(r,s)$ from $(q',p')$ using only $\varepsilon$ transitions. We then add an annotated variable transition $((q,p),(S,i),(r,s))$ to the
automaton. For instance, for the product automaton at the top of Figure
\ref{fig:numbered}, we would add a transition $((q_0,p_1),
(\Open{x},1),(q_4,p_2))$, due to the fact that we can reach $(q_1,p_1)$
from $(q_0,p_1)$ using $(\Open{x},1)$, and we can reach $(q_4,p_2)$ from
$(q_1,p_1)$ using $\varepsilon$ (which replaced $a$). We repeat this
procedure for all the variable transitions of $\cA\times \cA_d$, and
the newly added transitions, until no new transition can be
generated. In the end, we simply erase all the $\varepsilon$
transition from the resulting automaton. An example of this process
for the automaton $\cA$ of Figure~\ref{fig-automaton} and the document
$d=ab$ is given at the bottom of Figure \ref{fig:numbered}.

%In the final step, we replace all letter transitions with
%$\varepsilon$-transitions and compute what we call the ``forward''
%$\varepsilon$-closure. This is done by considering each variable
%transition $(q,S,q')$, and then computing all the states $p$ such that
%one can reach $p$ from $q'$ using only $\varepsilon$ transitions, and
%such that $p$ itself has at least one outgoing variable
%transition. We then add a variable transition $(q,S,p)$ to the
%automaton. For instance, for the product automaton at the top of Figure
%\ref{fig:numbered}, we would add a transition $((q_0,p_1),
%\Open{x},(q_4,p_2))$, due to the fact that we can reach $(q_1,p_1)$
%from $(q_0,p_1)$ using $\Open{x}$, and we can reach $(q_4,p_2)$ from
%$(q_1,p_1)$ using $\varepsilon$ (which replaced $a$). We repeat this
%procedure for all the variable transitions of $\cA\times \cA_d$, and
%the newly added transitions, until no new transition can be
%generated. In the end, we simply erase all the $\varepsilon$
%transition from the resulting automaton. An example of this process
%for the automaton $\cA$ of Figure~\ref{fig-automaton} and the document
%$d=ab$ is given at the bottom of Figure \ref{fig:numbered}.

From the resulting \dag we can now
easily enumerate $\semd{\cA}$. For this, we simply start from the final
state, and do a depth-first traversal taking all the edges
backwards. Every time we reach the initial state, we will have the
complete information necessary to construct one of the output
mappings. For example, starting from the accepting state and moving
backwards to $(q_3,p_3)$, and then again to the initial state. From
the labels along this run we can then reconstruct the mapping $\mu$
with $\mu(x)=\mu(y)=[1,3\rangle$.

Since $\cA$ and $\cA_d$ are deterministic, we will never output the same mapping twice. Also, note that the time for generating each output is bounded by the number of variables in $\cA$, and therefore the delay between outputs depends only on $|\cA|$ (and is constant in the document).

\subsubsection{The algorithm}  
%\medskip
%\noindent{\bf The algorithm.}
While the previous construction works
correctly, there is no need to perform the three construction phases
separately in a practical implementation. In fact, by a clever merge
of the three construction steps we can avoid materializing
$\cA_d$ and $\cA \times \cA_d$ altogether. The result is a succinct,
optimized, and easily-implementable algorithm that we describe next.

There are two main differences with the construction described above
and our algorithm. First, the algorithm never materializes $\cA_d$,
nor the product $\cA \times \cA_d$. Rather, it
\emph{traverses} this product automaton on-the-fly by processing the
input document one letter at a time. Second, the algorithm does not
construct the $\varepsilon$-closure itself, but its \emph{reverse
  dual}. That is, the resulting \dag has the edge labels of the
$\varepsilon$-closure as nodes and there is an edge from $(T,j)
\rightarrow (S,i)$ in the reverse dual if we had $(q,p)
\xrightarrow{(S,i)} (q',p') \xrightarrow{(T,j)} (q'', p'')$ in the
$\varepsilon$-closure for some product states $(q,p), (q',p')$, and
$(q'', p'')$. 
%To illustrate, the dashed arrows in Figure~\ref{fig-dag} show the reverse dual of the $\varepsilon$-closure shown in Figure~\ref{fig:epsilon-closure}. 

The algorithm builds the reverse dual \dag incrementally by processing
$d$ one letter at a time. In order to do this, it tracks at every
position $i$ ($1 \leq i \leq |d|)$ the states of $\cA$ that are
\emph{live}, i.e., the states $q \in Q$ such that there exists at least
one run of $\cA$, on the prefix $d(1,i)$ of $d$ that ends in $q$. For
each such state, the algorithm keeps track of the nodes in the reverse
dual that represent the last variable transitions taken by runs ending
in $q$. When appropriate, new nodes are added to the reverse dual
based on this information.

The different procedures that comprise the evaluation algorithm are
given in Algorithms~\ref{alg:ma-eval} and~\ref{alg:output}. In
particular, the procedure \textsc{Evaluate} shown in
Algorithm~\ref{alg:ma-eval} takes a deterministic and sequential $\EVA$ $\cA$ and a document $d = a_1 \ldots a_n$ as
input, and creates the reverse dual \dag that encodes all the runs of
$\cA$ over $d$. The procedure \textsc{Enumerate} shown in
Algorithm~\ref{alg:output} enumerates all the resulting mappings.
Before discussing these procedures in detail, we need to elaborate on
the data structures used.

\smallskip
\noindent{\bf Data structures.} We store the reverse dual \dag by
using the adjacency list representation. Each node $n$ in this \dag
is a pair $((S, i), l)$ where $S \subseteq \markers_{\VV}$, $i
\in \bbN$, and a $l$ is the list of nodes to which $n$ has outgoing
edges. Given a node $n$, the method $n.{\tt content}$ retrieves the
pair $(S,i)$ while the method $n.{\tt list}$ retrieves the adjacency
list $l$. A special node, denoted by $\bot$ will be used as the sink
node (playing the same role as the initial state of $\cA\times \cA_d$).

%\stijn{Here is a paragraph that I believe makes explicit the claims  wrt the running times of the different operations. It may be overly  detailed, but avoids confusion (I think). If you want to rever to  the old text, it is in comment.}

The algorithm makes extensive use of list operations. Lists are
represented as a pair $(s,e)$ of pointers to the start and end elements
in a singly linked list of elements. Elements are created and
never modified. The only exception to this is an element whose {\tt
  next} pointer is {\tt null}. Such an element may have its {\tt next}
pointer  updated, but only once. Lists are endowed with six
methods: {\tt begin}, {\tt next}, {\tt atEnd}, {\tt add}, {\tt
  lazycopy}, and {\tt append}. The first three methods {\tt begin},
{\tt next}, and {\tt atEnd} are standard methods for iterating through
a list.  Specifically, {\tt begin} starts the iteration from the  beginning (i.e. it locates the position {\em before} the first node),  {\tt next} gives the next node in the list, and {\tt atEnd} tells  whether the iteration is at the end or not.
The last three methods {\tt add}, {\tt lazycopy} and {\tt append} are
methods for modifying or extending a list $l = (s,e)$. {\tt add}
receives a node $n$ and inserts $n$ at the beginning of $l$ (i.e., it
creates a new element whose payload is $n$ and whose {\tt next}
pointer is $s$, and updates $l := (s',e)$ with $s'$ pointing to this
new element). {\tt lazycopy} makes a lazy copy of $l$ by returning a
copy of the pair $(s,e)$. This copy is not updated on operations to
$l$ (such as, {\tt add}, which would modify $s$).  {\tt append}
receives another list $l' = (s', e')$ and appends $l'$ at the end of
$l=(s,e)$ by updating the {\tt next} pointer of $e$ to $s'$ and subsequently
updating $l$ to $(s, e')$. Note that all of these operations are
clearly $O(1)$ operations.

\begin{algorithm}[t]
%	\caption{Evaluate $A = (Q, q_0, F, \delta)$ over a word $a_1 \ldots a_n$}\label{alg:ma-eval}
	\caption{Evaluate $\cA$ over the document  $a_1 \ldots a_n$}\label{alg:ma-eval}
	\begin{algorithmic}[1]
		\Procedure{Evaluate}{$\cA$, $a_1 \ldots a_n$}
			\ForAll{$q \in Q \setminus \{q_0\}$}
				\State $\plist_q \gets \epsilon$
			\EndFor
			\State $\plist_{q_0} \gets [\bot]$

			\For{$i := 1$ to $n$}
				\State $\textsc{Capturing}(i)$			
				\State $\textsc{Reading}(i)$
	
			\EndFor
			
			\State $\textsc{Capturing}(n+1)$

			\State $\textsc{Enumerate}(\{\plist_q\}_{q \in Q}, F)$
		\EndProcedure
		
		\medskip
		
		\Procedure{Capturing}{$i$}
			\ForAll{$q \in Q$}
				\State $\plistold_q \gets \plist_q \!.{\tt lazycopy}$
			\EndFor
			\ForAll{$q \in Q$ \textbf{with} $\plistold_q \neq \epsilon$}  					\ForAll{$S \in \markers_\delta(q)$}
					\State $\newnode \gets {\tt Node}((S,i),\plistold_q)$ \label{alg-node-creation}
					\State $p \gets \delta(q,S)$
					\State $\plist_{p}\!.{\tt add}(\newnode)$
				\EndFor
			\EndFor
		\EndProcedure
		
		\medskip
		
		\Procedure{Reading}{$i$}
		
			\ForAll{$q \in Q$}
				\State $\plistold_q \gets \plist_q$
				\State $\plist_q \gets \epsilon$
			\EndFor
			
			\ForAll{$q \in Q$ \textbf{with} $\plistold_q \neq \epsilon$}
				\State $p \gets \delta(q,a_i)$
				
				\State $\plist_{p}\!.{\tt append}(\plistold_{q})$
			\EndFor
		\EndProcedure
	\end{algorithmic}
\end{algorithm}

\begin{algorithm}[t]
	\caption{Enumerate all mappings}\label{alg:output}
	\begin{algorithmic}[1]
		\Procedure{{Enumerate}}{$\{\plist_q\}_{q \in Q}$, $F$}
		\ForAll{$q \in F$ \textbf{with} $\plist_q \neq \epsilon$}
		\State ${\tt EnumAll}(\plist_q,\epsilon)$
		\EndFor
		\EndProcedure
		
		\smallskip
		
		\Procedure{{EnumAll}}{$\plist, \newmap$}
		
		\State $\plist\!.{\tt begin}$
		
		\While{$\plist\!.{\tt atEnd} = {\tt false}$}
				\State $\newnode \gets \plist\!.{\tt next}$
				
				\If{$\newnode = \bot$}
					\State ${\tt Output}(\newmap)$
				\Else
					\State $(S, i) \gets \newnode.{\tt content}$
					\State $\textsc{EnumAll}(\newnode.{\tt list}, \,(S, i) \cdot \newmap)$
				\EndIf			
		\EndWhile
		\EndProcedure
	\end{algorithmic}
\end{algorithm}

\smallskip
\noindent{\bf Evaluation.} 
The procedure $\textsc{Evaluate}$ maintains a list $\plist_q$ of
nodes, for every state $q$ of $\cA$. If $\plist_q$ is empty, then $q$
is not live for the current letter position. Otherwise, $q$ is live
and $\plist_q$ contains the nodes in the reverse dual \dag that
represent the last variable transitions taken by runs of $\cA$ on the
current prefix that end in $q$. Initially, $\plist_q$ is empty for
every state $q$ except the initial state $q_0$, which is initialized
to the singleton list containing the special sink node
$\bot$. \textsc{Evaluate} then alternates between calls to
$\textsc{Capturing}(i)$ and $\textsc{Reading}(i)$, where $i$ is a
letter position in $d$ (recall that all the runs of an extended automata alternate between variable and
letter transitions and start with a variable transition,
cf. \eqref{eq:run}). $\textsc{Capturing}(i)$ simulates the variable
transitions that $\cA$ does immediately before reading the letter
$a_i$, and modifies the reverse dual \dag accordingly. Similarly, $\textsc{Reading}(i)$ simulates what
$\cA$ does when reading the letter $a_i$ of the input. Finally,
$\textsc{Capturing}(n+1)$ simulates the last variable transition of
$\cA$.

In $\textsc{Capturing}(i)$ we first make a lazy copy of all the
lists. We then try to extend the runs of $\cA$ from each state $q$
that was live at position $i-1$ (i.e., $\plistold_q\neq \epsilon$) by
executing a variable transition. If we can do this (i.e. there is a
transition of the form $(q,S,p)$ in $\cA$), we create a new node $n$
labeled by $(S,i)$ that has an edge to each node in
$\plist_q^{old}$. Finally, we add $n$ to the beginning of the list
$\plist_p$, thus recording that $\cA$ can be in state $p$ after
executing the $i$th variable transition. Notice that it is possible
that two transitions enter the same state $p$ (like the transitions
reaching the accepting state in Figure \ref{fig-automaton}). To
accommodate for this, our algorithm adds the new node at the beginning
of the list, so by traversing the entire list we get the information
about all the runs.

It is important to note that in $\textsc{Capturing}(i)$ we do not
overwrite the lists $\plist_q$ that were created in
$\textsc{Reading}(i-1)$ for $i > 1$. This is necessary to correctly keep track of the situation in which no transition using variable markers was triggered in $\textsc{Capturing}(i)$ (i.e. when $S=\emptyset$ in our run). On a run of a sequential extended variable-set automaton this can happen for instance when we have self loops (as in e.g. state $q_3$ in Figure \ref{fig-automaton}). This way, the list $\plist_q$ is kept for the next iteration; i.e. $\textsc{Reading}(i)$ can again continue from $q$ since no variable markers were used in between.

In $\textsc{Reading}(i)$ we simulate what happens when $\cA$ reads the
letter $a_i$ of the input document by updating the lists of the states
that $\cA$ reaches in this transition. That is, we first mark all
lists as ``old'' lists, and then set $\plist_q$ to empty. Then for
each live state $q$ (i.e., $\plist_q^{old}\neq \epsilon$, hence $\cA$
was in $q$ immediately before reading $a_i$), and the transitions of
the form $(q,a_i,p)$, we append the list $\plist_q^{old}$ at the end
of the list $\plist_p$. Appending this list at the end is done in
order to accommodate the fact that two letter transitions can enter
the same state $p$ while reading $a_i$ (see e.g. the state $q_8$ in
the automaton from Figure \ref{fig-automaton}).  Note here that each
$\plist_q^{old}$ is appended to at most one $\plist_p$, since $\cA$ is
deterministic.

\smallskip
\noindent{\bf Enumeration.} At the end of \textsc{Evaluate}, procedure
\textsc{Enumerate} simply traverses the constructed reverse dual \dag
in a depth first manner. In this way, \textsc{Enumerate} traces all
the accepting runs (since it starts from an accepting state), and
outputs a string allowing us to reconstruct the mapping.

\medskip
\noindent{\bf Example.} Next we give an example detailing the situations that could occur while running Algorithm \ref{alg:ma-eval}. For this, consider the deterministic $\sEVA$ $\cA$ from Figure \ref{fig-automaton} and an input document $d=ab$. In this case we have that $\semd{\cA}=\{\mu_1,\mu_2,\mu_3\}$, where: 
\begin{itemize}\itemsep=0pt
\item $\mu_1(x)  =  [1,3\rangle, \ \mu_1(y)  =  [2,3\rangle$;
\item $\mu_2(x)  =  [2,3\rangle, \ \mu_2(y)  =  [1,3\rangle$; and
\item $\mu_3(x)  =  [1,3\rangle, \ \mu_3(y)  =  [1,3\rangle$.
\end{itemize}
%\begin{align*}
%\mu_1(x)  =  (1,3) & , & \mu_1(y)  =  (2,3)\\
%\mu_2(x)  =  (2,3) & , & \mu_2(y)  =  (1,3)\\
%\mu_3(x)  =  (1,3) & , & \mu_3(y)  =  (1,3)
%\end{align*}

To show how Algorithm \ref{alg:ma-eval} works, in Figure \ref{fig-lists} we provide the state of all the active lists after completion of each phase of the algorithm. To stress that we are talking about the state of some list $\plist_q$ during the iteration $i$ of Algorithm \ref{alg:ma-eval}, that is, about the state of the list after executing $\textsc{Reading}(i)$ or $\textsc{Capturing}(i)$, we will use the notation $\plist_q^i$. To keep the notation simple, we also denote lists using the array notation.

\begin{figure}[t]
\centering
\resizebox{\linewidth}{!}{
\begin{tabular}{c|l}
 Stage & \hspace{45pt}Non-empty lists \\
 \hline
 Initial & 
 \def\arraystretch{1.5}
 \begin{tabular}{l}
 $\plist_{q_0}^0 = [\bot]$
 \end{tabular}\\
 \hline
 $\textsc{Capturing}(1)$ &  
 \def\arraystretch{1.5}
 \begin{tabular}{l}
 $\plist_{q_0}^0 = [\bot]$\\
 $\plist_{q_1}^0 = [\texttt{node(}(\{\Open{x}\},1),[\bot]\texttt{)}]$\\
 $\plist_{q_2}^0 = [\texttt{node(}(\{\Open{y} \},1),[\bot]\texttt{)}]$\\
 $\plist_{q_3}^0 = [\texttt{node(}(\{\Open{x},\Open{y}\},1),[\bot]\texttt{)}]$
 \end{tabular}\\
 \hline
 $\textsc{Reading}(1)$ & 
 \def\arraystretch{1.5}
 \begin{tabular}{l}
 $\plist_{q_4}^1 = \plist_{q_1}^0$\\
 $\plist_{q_5}^1 = \plist_{q_2}^0$\\
 $\plist_{q_3}^1 = \plist_{q_3}^0$
 \end{tabular}\\
 \hline
 $\textsc{Capturing}(2)$ &
 \def\arraystretch{1.5}
 \begin{tabular}{l}
 $\plist_{q_4}^1 = \plist_{q_1}^0$\\
 $\plist_{q_5}^1 = \plist_{q_2}^0$\\
 $\plist_{q_3}^1 = \plist_{q_3}^0$\\
 $\plist_{q_6}^1 = [\texttt{node(}(\{\Open{y} \},2),\plist_{q_4}^1\texttt{)}]$\\
 $\plist_{q_7}^1 = [\texttt{node(}(\{\Open{x} \},2),\plist_{q_5}^1\texttt{)}]$\\
 $\plist_{q_9}^1 = [\texttt{node(}(\{\Close{x},\Close{y} \},2),\plist_{q_3}^1\texttt{)}]$
 \end{tabular} 
  \\ \hline
 $\textsc{Reading}(2)$ &  
  \def\arraystretch{1.5}
 \begin{tabular}{l}
 $\plist_{q_3}^2 = \plist_{q_3}^1$\\
 $\plist_{q_8}^2 = [\plist_{q_6}^1,\plist_{q_7}^1]$
 \end{tabular} 
 \\ \hline
 $\textsc{Capturing}(3)$ &
   \def\arraystretch{1.5}
 \begin{tabular}{l}
 $\plist_{q_3}^2 = \plist_{q_3}^1$\\
 $\plist_{q_8}^2 = [\plist_{q_6}^1,\plist_{q_7}^1]$\\
 $\plist_{q_9}^2 = [\texttt{node(}(\{\Close{x},\Close{y}\},3),\plist_{q_8}^2\texttt{)},$ \\
 \hspace*{35pt} $\texttt{node(}(\{\Close{x},\Close{y}\},3),\plist_{q_3}^2\texttt{)}]$
 \end{tabular}  
\end{tabular}
}
\caption{The state of non-empty lists after executing each stage of the algorithm.}
\label{fig-lists}
\end{figure}

At the beginning only the list $\plist_{q_0}$ corresponding to the
initial state of $\cA$ is non-empty. When $\textsc{Capturing}(1)$ is
triggered, we create three new nodes, each corresponding to the
variable transitions leaving the state $q_0$. These nodes are then
added to the appropriate lists. In $\textsc{Reading}(1)$ we ``move"
the non-empty lists by renaming their state. For instance, since $A$
can go from $q_1$ to $q_4 $ while reading $a_1=a$, the list
$\plist_{q_1}^0$ now becomes $\plist_{q_4}^1$, signalling that $q_4$
is one of the states where $\cA$ can be at this point. The same is done
by the other two transition reading the letter $a$. Notice that the
list $\plist_{q_0}$ becomes empty at this point.

Next, $\textsc{Capturing}(2)$ is executed. Here, the lists that were
non-empty after $\textsc{Reading}(1)$ will remain unchanged after
$\textsc{Capturing}(2)$, simulating the situation when no variable
bindings were used in the run of $\cA$ over $d$ after processing the
first letter. Other variable transitions that can be triggered create
new nodes and add them at the beginning of the appropriate lists.

$\textsc{Reading}(2)$ again ``moves" the lists according to what $\cA$ does when reading $a_2=b$. The lists $\plist_{q_3}^1$ gets propagated (simulating a self loop). A more interesting situation occurs when the transitions $\delta(q_6,b)=q_8$ and $\delta(q_7,b)=q_8$ are processed. Since they both reach $q_8$, we first append the list $\plist_{q_6}^1$ at the end of (the empty list) $\plist_{q_8}^2$, and then to keep track that one can also get here from $q_7$, also append the list $\plist_{q_7}^1$ at the end of (now non empty list) $\plist_{q_8}^2$. Since these are the only way that $\cA$ can move while reading $b$, we forget about all the other lists.

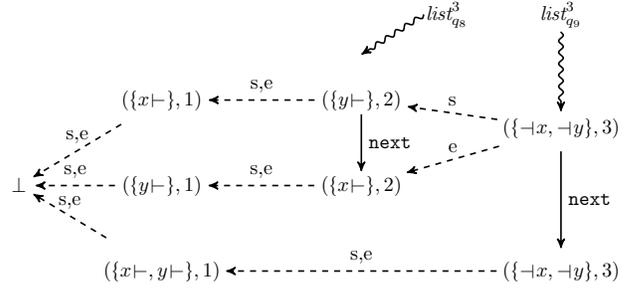
\begin{figure}[t]
\begin{center}
\resizebox{\linewidth}{!}{
\begin{tikzpicture}[->,>=stealth',auto,thick, scale = 1.0,state/.style={circle,inner sep=2pt}]

    % The graph
	\node [state] at (0,0) (b) {$\bot$};
	\node [state] at (2.5,1.5) (x0){$(\{\Open{x}\},1)$};  
	\node [state] at (2.5,0) (y0){$(\{\Open{y}\},1)$};  	
	\node [state] at (2.5,-1.5) (xy0){$(\{\Open{x},\Open{y}\},1)$};  	
	\node [state] at (6,1.5) (y1){$(\{\Open{y}\},2)$};
	\node [state] at (6,0) (x1){$(\{\Open{x}\},2)$};
	\node [state] at (9.5,1) (xy2){$(\{\Close{x},\Close{y}\},3)$};
	\node [state] at (9.5,-1.5) (xy22){$(\{\Close{x},\Close{y}\},3)$}; 

	% Graph edges
	\path[dashed,->]
	(x0) edge node[above] {s,e}   (b)
	(y0) edge node[above] {s,e}  (b)
	(xy0) edge node[above] {s,e}   (b)
	(y1) edge node[above] {s,e}  (x0)	
	(x1) edge node[above] {s,e}  (y0)
	(xy2) edge[above] node {s} (y1)
	(xy2) edge[above] node {e}  (x1)			
	(xy22) edge node[above] {s,e}  (xy0)	
	;  	
	
	\path[->]	
	(6,1.25) edge node[right] {\texttt{next}} (6,0.25)
	(9.5,0.6) edge node[right] {\texttt{next}} (9.5,-1.1)
	;
	
	\node[state] at (9.5,3) {$\plist_{q_9}^3$};
	\node[state] at (7.5,3) {$\plist_{q_8}^3$};
		
	\draw [->,decorate,decoration={snake,amplitude=.4mm,segment length=2mm,post length=1mm}] (7.1,3) -- (6,2.3);
	
	\draw [->,decorate,decoration={snake,amplitude=.4mm,segment length=2mm,post length=1mm}] (9.5,2.7) -- (9.5,1.3);

\end{tikzpicture} 
}
\end{center}
\vspace*{-20pt}
\caption{DAG created by Algorithm \ref{alg:ma-eval} to record the output mappings.}
\label{fig-dag}
\end{figure}

Finally, $\textsc{Capturing}(3)$ keeps track of what happens during the last variable transition of $\cA$. There are two transitions that can reach the accepting state $q_9$, and they get added to the list $\plist_{q_9}^3$. Note that the two lists from $\textsc{Reading}(2)$ also remain non-empty at this stage.

The \dag created by Algorithm \ref{alg:ma-eval} is given in Figure~\ref{fig-dag}. Here the dashed edges point to the list corresponding to the node with this label (i.e. the list representation $(s,e)$). For instance, $\texttt{node(}(\{\Open{x}\},0),\bot\texttt{)}$ corresponds to the edge between the node with the label $(\{\Open{x}\},0)$ and its associated list $\bot=\plist_{q_0}^0$. Full edges link the nodes that belong to the same list, and curvy edges to the start of a list generated after $\textsc{Capturing}(3)$. 

To enumerate the answers, we now call the procedure $\textsc{Enumerate}$, passing it as a parameter all the lists corresponding to the final states of $\cA$. Since $\cA$ has only one final state, the procedure will trigger only $\textsc{EnumAll}(\plist_{q_9},\varepsilon)$. This procedure now recursively traverses the structure of connected lists created by Algorithm \ref{alg:ma-eval} in  a depth-first manner generating the output mappings. For instance, the mapping $\mu_1$, with $\mu_1(x)=[1,3\rangle$ and $\mu_2(y)=[2,3\rangle$ is generated by traversing the upper most path from $(\{\Close{x},\Close{y}\},2)$ until reaching $\bot$, and similarly for other mappings.

\medskip
\noindent{\bf Correctness.} To prove the correctness of the above algorithm, we first introduce some notation. 
For encoding mappings in the enumeration procedure, we assume that mappings are sequences of the form $(S_1, i_1)$ \ldots $(S_m, i_m)$ where $S_j \subseteq \markers_{\VV}$, $i_1 < \ldots < i_m$ and variables in $S_1 \ldots S_m$ are open and closed in a correct manner, i.e. like in the definition of a run of an extended variable set automata. Clearly, from a sequence $M = (S_1, i_1) \ldots (S_m, i_m)$ we can obtain a mapping $\mu^M$ and viceversa. 
For this reason, in the sequel we call $M$ and $\mu$ mappings without making any distinction.
Furthermore, we say that a sequence $M = (S_1, i_1) \ldots (S_k, i_k)$ is a \emph{partial mapping} if it is the prefix sequence of some mapping, i.e., it can be extended to the right to create a mapping. 
This is useful to represent the output of partial  
run of $\cA$ over $d$; that is, if $\rho$ is of the form:
$$\rho \ = \ q_0 \ \trans{S_1} \ p_0 \ \trans{a_1} \ q_1 \ \trans{S_2} \ p_1 \ \trans{a_2} \ \ldots \ \ \trans{a_i} \ q_i \ \trans{S_{i+1}} {p_{i}}$$
where $i\leq |d|$, the mapping $\mu^{\rho}$ is not necessarily
well-defined, or is possibly incomplete. We therefore define a partial
mapping $M$ of $\rho$, denoted by $\Out(\rho)$, as the concatenation
of all the pairs $(S_j,j)$ where $S_j\neq \emptyset$, in an 
%\stijn{Did  you men increasing order here? If not, the example is wrong.}
increasing order on $j$. For instance, in the run $\rho\ = \ q_0 \ \trans{\{\Open{x}\}} \ p_0 \ \trans{a_1} \ q_1 \ \trans{\emptyset} p_1 \ \trans{a_2} \ q_2 \ \trans{\{\Open{y}\}} p_2$ we will have that $\Out(\rho)= (\{\Open{x}\},0) \, (\{\Open{y}\},2)$. Note that in the case that $\rho$ is an accepting run of $A$, it is then clear that $\Out(\rho)$ defines the mapping $\mu^\rho$.

The proof that Algorithm \ref{alg:ma-eval} correctly enumerates all the mappings in $\semd{\cA}$ without repetitions follows from  the invariant stated in the lemma below.

\begin{lemma}\label{lemma:correctness1}
	Let $d = a_1 \ldots a_n$ be a document and $\cA$ an extended variable-set automaton that is deterministic and sequential (deterministic $\sEVA$). Then for every $0\leq i \leq n$, the following two statements are equivalent:
	\begin{enumerate}
		\item There exists a run of $\cA$ over $a_1\cdots a_i$ of the form $\rho \ = \ q_0 \ \trans{S_1} \ p_0 \ \trans{a_1} \ q_1 \ \ldots \ \ \trans{a_i} \ q_{i} \ \trans{S_{i+1}} p_i.$ % over $a_1 \ldots a_i$.
		\item After executing $\textsc{Capturing}(i+1)$ in Algorithm \ref{alg:ma-eval}, it holds that $\plist_{p_i} \neq \epsilon$ and there is partial output $M$ of \textsc{EnumAll}($\plist_{p_i},\epsilon$) with $M = \Out(\rho)$.
	\end{enumerate}
\end{lemma}

The proof of the lemma is done by a detailed induction on the number of steps of the algorithm and can be found in the appendix. Note that the case when $i=0$ corresponds to a run over the empty word $\varepsilon$ (i.e. processing the part of $d$ ``before" $a_1$), thus simulating the first variable transition of $\cA$. With the invariant proved in Lemma \ref{lemma:correctness1}, we can now easily show that running $\textsc{Evaluate}(\cA,d)$ will enumerate all of the mappings in $\semd{\cA}$ and only those mappings. Indeed, if $\mu\in \semd{\cA}$, this means that there is an accepting run $\rho$ such that $\mu^{\rho} = \mu$, so by Lemma \ref{lemma:correctness1}, the algorithm will output $M$ with $M=\Out(\rho)$. On the other hand, if $\textsc{Evaluate}(\cA,d)$ produces an output $M$, we can match this output with a run $\rho_M$. Furthermore, since the output was produced from an accepting state, and since $\cA$ is sequential, this means $\rho_M$ is valid, so $\mu^{\rho_M}=\mu^M\in \semd{A}$ as desired.

%This now allows us to prove that Algorithm \ref{lemma:correctness1} indeed enumerates all of the mappings in $\semd{A}$ and only those mappings. On the one hand, if $\mu\in \semd{A}$, this means that there is an accepting run $\rho$ such that $\mu^{\rho} = \mu$. By Lemma \ref{lemma:correctness1}, this means that Algorithm \ref{alg:ma-eval} will produce an output $M$ such that $M = \Out(\rho)$, so we have that $\mu^M=\mu$. Conversely, if $M$ is an output of Algorithm \ref{alg:ma-eval}, this means that $M$ was obtained by running \textsc{EnumAll}($\newnode.{\tt list}, (S, i) \cdot \epsilon$), for some $\newnode$ on the list $\plist_{q}^{n}$ and $\newnode.{\tt content} = (S,i)$, with $q$ and accepting state of $A$, and $n=|d|$. Again, by Lemma \ref{lemma:correctness1}, this means that there is a run $\rho$ ending in $q$, such that $\Out(\rho) = M$. Since $\rho$ ends in $q$ it is accepting, and since $A$ is functional, we know that $\rho$ is valid. Therefore $\mu^\rho$ is well defined and  $\mu^\rho \in \semd{A}$, so $\mu^M\in \semd{A}$ as desired.

Finally, we need to show that Algorithm \ref{alg:ma-eval} does not enumerate any answer twice when executed over a deterministic $\sEVA$ $\cA$ and a document $d$. For this, observe  that if we have two accepting runs $\rho$ and $\rho'$ of $\cA$ over $d$ such that $\mu^\rho=\mu^{\rho'}$, then $\rho=\rho'$. This follows from the fact that $\cA$ is deterministic. Therefore, it follows from Lemma \ref{lemma:correctness1} that there is a one to one correspondence between accepting runs of $A$ and outputs of Algorithm~\ref{alg:ma-eval}, which gives us the desired result.

\medskip
\noindent{\bf Complexity.} It is rather straightforward to see that the pre-processing step takes time $O(|\cA|\times |d|)$. Namely, for each letter $a_i$ of $d$ we run the procedures $\textsc{Capturing}(i)$ and $\textsc{Reading}(i)$ once. These two procedures simply scan the transitions of the automaton and manipulate the list pointers as needed, thus taking $O(|\cA|)$ time, where $|\cA|$ is measured as the number of transitions of the automaton. 

As far as the enumeration is concerned, Algorithm \ref{alg:output}, traverses the graph generated in the pre-processing step in a depth-first manner. From Lemma \ref{lemma:correctness1}, it follows that all the paths in the constructed graph must reach the initial node $\bot$ and that the length of each path is linear in the number of variables. Thus, we are able to enumerate the output by taking only constant delay (i.e. constant in the size of the document) between two consecutive mappings.

Note that the actual delay is not really dependent on the entire automaton $\cA$, as allowed by the definition of constant delay, but depends only on the number of variables. We argue that this is the best delay that can be achieved, because to write down a single output mapping one needs at least the time that is linear in the number of variables.
%
%As far as the enumeration is concerned, Algorithm \ref{alg:output}, traverses the graph constructed in the pre-processing step in a depth-first manner, and is thus not constant delay, as one might need to backtrack all the way to the final state. However, from Lemma \ref{lemma:correctness1}, it follows that all the paths in the constructed graph must reach the node $\bot$, thus we can output them taking only constant delay between consecutive outputs by implementing  ``smart" backtracking that remembers the ultimate fork in the path using a stack. For space reasons, and in order to keep the presentation more understandable, we present the constant delay version of Algorithm \ref{alg:output} in the appendix.

%%% Local Variables: 
%%% mode: latex
%%% TeX-master: "../main/main"
%%% End: 

%% file: spanners1.tex
%!TEX root = ../main/main.tex

The previous section shows an algorithm that evaluates a deterministic
and sequential extended VA (\dsEVA\ for short) $\cA$ over a document $d$ with
constant-delay enumeration after $O(|\cA|\times |d|)$
preprocessing. Since the wider objective of this algorithm is to
evaluate regular spanners, in this section we present a fine-grained
study of the complexity of transforming an arbitrary regular spanner,
expressed in $\RGX^{\{\pi,\cup,\Join\}}$ or
$\VA^{\{\pi,\cup,\Join\}}$ to a $\dsEVA$. This will illustrate the real
cost of our constant delay algorithm for evaluating
regular spanners.

Because it is well-known that $\RGX$ formulas can be translated into $\VA$ in
linear time~\cite{FKRV15}, we can focus our study on the
setting where spanners are expressed in $\VA^\algops$. We first
study how to translate arbitrary $\VA$s into $\dsEVA$s, and then turn to the
algebraic constructs.  For the sake of simplification, throughout this section
we assume the following notation: given a VA $\cA = (Q, q_0, F,
\delta)$, $n = |Q|$ denotes the number of states, $m = |\delta|$ the
number of transitions, and $\ell = |\var(\cA)|$ the number of
variables in $\cA$.

To obtain a sequential $\VA$ from a $\VA$, we can use a construction similar to the one presented in~\cite{F17}. This yields a sequential $\VA$ with $2^n3^\ell$ states that can later be extended and determinized (see Theorem~\ref{theo:equivalence} and Proposition~\ref{prop:determinization}, respectively). Unfortunately, following these steps would yield an automaton whose size is double exponential in the size of the original $\VA$. The first positive result in this section is that we can actually transform a $\VA$ into a $\dsEVA$ avoiding this double exponential blow-up.
\begin{proposition}\label{prop:general-VA}
	For any $\VA$ $\cA$ there exists an equivalent $\dsEVA$ $\cA'$ with at most $2^n3^\ell$ states and $2^n3^\ell(2^\ell+|\Sigma|)$ transitions.
\end{proposition}
Therefore, evaluating an arbitrary $\VA$ with constant delay can be done with preprocessing that is exponential in the size of the $\VA$ and linear in the document. However, note that the resulting $\dsEVA$ is exponential both in the number of states and in the number of variables of the original~$\VA$. While having an automaton that is exponential in the number of states is to be expected due to the deterministic restriction of the resulting $\VA$, it is natural to ask whether there exists a  subclass of $\VA$ where the blow-up in the number of variables can be avoided. 

The two subclasses of $\VA$ that were shown to have good algorithmic properties \cite{FKP17,MaturanaRV17} are sequential $\VA$ and functional $\VA$, so we will consider if the cost of translation is smaller in these cases. In the more general case of sequential $\VA$ we can actually show that the blow-up in the number of variables is inevitable. The main issue here is that preserving the sequentiality of a $\VA$ when transforming it to an extended $\VA$ can be costly. 
To illustrate this, consider the automaton in
Figure~\ref{fig:badsequential}. In this automaton any path between $q_0$ and $q_F$ opens and closes exactly one
variable in $\{x_i,y_i\}$, for each $i\in\{1,\ldots,n\}$. Therefore,
to simulate this behaviour in an extended \VA\ (which disallows two
consecutive variable transitions), we need $2^{\ell}$ transitions between
the initial and final states, one for each possible set of variables. More formally, we have the following proposition.

 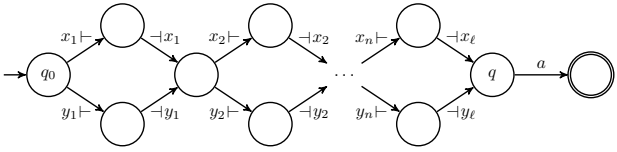
\begin{figure}
 	\begin{center}
 		\resizebox{\columnwidth}{!}{
 			\begin{tikzpicture}[->,>=stealth',auto, thick, scale = 1.0,initial text= {}]
			
 			% The graph
 			\node [state,initial] at (0,0) (q0) {$q_0$};
 			\node [state] at (1.5,1) (q01) {};
 			\node [state] at (1.5,-1) (q02) {};		  
 			\node [state] at (3,0) (q1) {};
 			\node [state] at (4.5,1) (q11) {};
 			\node [state] at (4.5,-1) (q12) {};
 			\node [inner sep=2mm] at (6,0) (q2) {$\ldots$};
 			\node [state] at (7.5,1) (qn1) {};
 			\node [state] at (7.5,-1) (qn2) {};
 			\node [state] at (9,0) (qn) {$q$};
 			\node [state,accepting] at (11,0) (qf) {};

 			%\node [state,accepting] at (10,0) (q5) {$q_5$};
			
 			% Graph edges
 			\path[->] (q0) edge node[above] {$\Open{x_1} \ \ \ $} (q01);
 			\path[->] (q0) edge node[below] {$\Open{y_1} \ \ \ $} (q02);
 			\path[->] (q01) edge node[above] {$\ \ \Close{x_1}$} (q1);
 			\path[->] (q02) edge node[below] {$\ \ \Close{y_1}$} (q1);	  
 			\path[->] (q1) edge node[above] {$\Open{x_2} \ \ \ $} (q11);
 			\path[->] (q1) edge node[below] {$\Open{y_2} \ \ \ $} (q12);
 			\path[->] (q11) edge node[above] {$\ \  \Close{x_2}$} (q2);
 			\path[->] (q12) edge node[below] {$ \ \ \Close{y_2}$} (q2);
 			\path[->] (q2) edge node[above] {$\Open{x_n} \ \ \ $} (qn1);
 			\path[->] (q2) edge node[below] {$\Open{y_n} \ \ \ $} (qn2);
 			\path[->] (qn1) edge node[above] {$\ \ \Close{x_{\ell}}$} (qn);
 			\path[->] (qn2) edge node[below] {$\ \ \Close{y_{\ell}}$} (qn);
 			\path[->] (qn) edge node[above] {$a$} (qf);
 			\end{tikzpicture}
 		}
 	\end{center}
 	\vspace*{-15pt}
 	\caption{A sequential VA with $\ell$ variables such that every equivalent  eVA has $O(2^{\ell})$ transitions.}
 	\label{fig:badsequential}
 \end{figure}

\begin{proposition}\label{prop:seq-extended-existance}
For every $\ell>0$ there is a sequential VA $\cA$ with $3\ell+2$ states, $4\ell+1$ transitions, and $2\ell$ variables, such that for every extended $\VA$ $\cA'$ equivalent to $\cA$ it is the case that $\cA'$ has at least $2^{\ell}$ transitions.
\end{proposition}

On the other hand, if we consider functional $\VA$, the exponential factor depending on the number of variables can be eliminated when translating a functional $\VA$ into a $\dsEVA$.

\begin{proposition}\label{prop:func-extended-blowup}
	For any functional $\VA$ $\cA$ there exists an equivalent $\dsEVA$ $\cA'$ with at most $2^n$ states and $2^n(n^2+|\Sigma|)$ transitions.
\end{proposition}

Due to this, and the fact that functional $\VA$ are probably the class of $\VA$ most studied in the literature \cite{FKRV15,FKP17,F17}, for the remainder of this section we will be working with functional $\VA$.

Now we proceed to study how to apply the algebraic operators to evaluate regular spanners. In~\cite{FKRV15}, it was shown that any regular spanner (i.e. a join-union-projection expression built from $\RGX$ or $\VA$ as atoms) is in fact equivalent to a single $\VA$, and effective constructions were given. In particular, it is known that for every pair of $\VA$ $\cA_1$ and $\cA_2$, there exists a $\VA$ $\cA$ of exponential size such that $\semd{\cA} = \semd{\cA_1} \bowtie \semd{\cA_2}$. The exponential blow-up comes from the fact that each transition is equipped with at most one variable, and two variable transitions can occur consecutively. Therefore, one needs to consider all possible orders of consecutive variable transitions when computing a product (see \cite{FKRV15}). On the other hand, as shown by a subset of the author's in their previous work \cite{MorcianoUV16}, and independently in \cite{FKP17}, this blow-up can be avoided when working with functional $\VA$. In the next proposition, we generalize this result to extended $\VA$\footnote{Note that since \cite{FKP17} does not consider extended $\VA$,  the size of the join automaton is $O(n^4)$, and not quadratic.}.
\begin{proposition}\label{prop:algebra-eVA}
Let $\cA_1$ and $\cA_2$ be two functional $\EVA$, and $Y\subset \VV$. Furthermore, let $\cA_3$ and $\cA_4$ be two functional $\EVA$s that use the same set of variables. Then there exist functional $\EVA$s $\cA_{\bowtie}$, $\cA_{\cup}$, and $\cA_{\pi}$ such that:
\begin{itemize}\itemsep=0pt
\item $\cA_{\bowtie} \equiv \cA_1 \bowtie \cA_2$, and $\cA_{\bowtie}$ is of size $|\cA_1|\times |\cA_2|$.
\item $\cA_{\cup} \equiv \cA_3 \cup \cA_4$, and $\cA_{\cup}$ is of size $|\cA_3|+|\cA_4|$.
\item $\cA_{\pi} \equiv \pi_Y \cA_1$, and $\cA_{\pi}$ is of size $|\cA_1|$.
\end{itemize}
\end{proposition}
%\begin{proposition}\label{prop:algebra-eVA}
%Let $\cA_1$ and $\cA_2$ be two functional $\EVA$, and $Y\subset \VV$. Then there exist functional $\EVA$s $\cA_{\bowtie}$ and $\cA_{\pi}$ such that:
%\begin{itemize}\itemsep=0pt
%\item $\cA_{\bowtie} \equiv \cA_1 \bowtie \cA_2$, and $\cA_{\bowtie}$ is of size $|\cA_1|\times |\cA_2|$.
%\item $\cA_{\pi} \equiv \pi_Y \cA_1$, and $\cA_{\pi}$ is of linear size.
%\end{itemize}
%Furthermore, if $\cA_1$ and $\cA_2$ use the same set of variables, then there exists a functional $\EVA$ $\cA_{\cup}$ such that:
%\begin{itemize}
%\item $\cA_{\cup} \equiv \cA_1 \cup \cA_2$, and $\cA_{\cup}$ is of linear size.
%\end{itemize}
%\end{proposition}

Combining these results we can now determine the precise cost of compiling a regular spanner $\gamma$ into a $\dsEVA$ automaton that can then be used by the algorithm from Section \ref{sec:algorithm} to enumerate $\semd{\gamma}$ with constant delay, for an arbitrary document $d$. More precisely, we have the following.

\begin{proposition}\label{prop:reg_span_compiler}
Let $\gamma$ be a regular spanner in $\VA^{\{\pi,\cup,\Join\}}$ using $k$ functional $\VA$ as input, each of them with at most $n$ states. Then there exists an equivalent $\dsEVA$ $\cA_{\gamma}$ with at most $2^{n^k}$ states, and at most $2^{n^k}\cdot (n^{2k} + |\Sigma|)$ transitions.
\end{proposition}

In this case the $2^n$ factor from Proposition \ref{prop:func-extended-blowup} turns to $2^{n^k}$, thus making it double-exponential depending on the number of algebraic operations used in $\gamma$. Ideally, we would like to isolate a subclass of regular spanners for which this factor can be made single exponential. Unfortunately, in the general case we do not know if the double exponential factor $2^{n^k}$  can be avoided. The main problem here is dealing with projection, since it does not preserve determinism, thus causing an additional blow-up due to an extra determinization step. However, if we consider $\VA^{\{\cup,\bowtie\}}$, we can obtain the following.

%\begin{proposition}\label{prop:func_algebra_compilation}
%Let $\gamma$ be a regular spanner in $\VA^{\{\cup,\Join\}}$ mentioning $\ell$ variables, using $k$ algebraic operations and only functional $\VA$s with at most $n$ states. Then, there exists an equivalent $\dsEVA$ $\cA_{\gamma}$ with at most $2^{n\cdot k}$ states, at most $2^{n\cdot k}\cdot (n^{2k} + |\Sigma|)$ transitions and $\ell$ variables.
%\end{proposition}

\begin{proposition}\label{prop:func_algebra_compilation}
Let $\gamma$ be a regular spanner in $\VA^{\{\cup,\Join\}}$ using $k$ functional $\VA$ as input, each of them with at most $n$ states. Then, there exists an equivalent $\dsEVA$ $\cA_{\gamma}$ with at most $2^{n\cdot k}$ states, at most $2^{n\cdot k}\cdot (n^{2k} + |\Sigma|)$ transitions.
\end{proposition}

Overall, compiling arbitrary $\VA$ or expressions in $\VA^{\{\pi,\cup,\Join\}}$ into $\dsEVA$ can be quite costly. However, restricting to the functional setting and disallowing projections yields a class of document spanners where the size of the resulting $\dsEVA$ is manageable. In terms of practical applicability, it is also interesting to note that all of these translations can be fed to Algorithm \ref{alg:ma-eval} on-the-fly, thus rarely needing to materialize the entire $\dsEVA$.

% %%% Local Variables: 
% %%% mode: latex
% %%% TeX-master: "../main/main"
% %%% End: 

%% file: lowerbounds.tex
%!TEX root = ../main/main.tex

In this section we study the problem of counting the number of output mappings in $\semd{\gamma}$, where $\gamma$ is a document spanner. Counting the number of outputs is strongly related to the enumeration problem~\cite{Segoufin13} and can give some evidence on the limitations of finding constant delay algorithms with better precomputation phases. Formally, given a language
$\LL$ for specifying document spanners, we consider the following problem:
\begin{center}
	\framebox{
		\begin{tabular}{rl}
			\textbf{Problem:} & $\COUNT[\LL]$\\
			\textbf{Input:} &  An expression $\gamma \in \LL$, a document $d$.  \\
			\textbf{Output:} & $|\semd{\gamma}|$ 
		\end{tabular}
	}
\end{center}
It is common that constant delay enumeration algorithms can be extended to count the number of outputs efficiently~\cite{Segoufin13}. We show that this is the case for our algorithm over deterministic $\sEVA$. 

\begin{theorem}\label{theo:counting}
	Given a deterministic sequential extended $\VA$ $\cA$ and a document $d$, $|\semd{\gamma}|$ can be computed in time $O(|\cA|\times |d|)$.
\end{theorem}

Therefore, $\COUNT[\LL_1]$, where $\LL_1$ is the class of deterministic $\sEVA$, can be computed in polynomial time in combined complexity.
The algorithm for $\COUNT[\LL_1]$ can be found in the appendix. This algorithm is a direct extension of Algorithm~\ref{alg:ma-eval}, modified to keep the number of (partial) output mappings in each state instead of a compact representation of the mappings (i.e. $\plist_q$).

Unfortunately, the efficient algorithm of Theorem~\ref{theo:counting} cannot be extend beyond the class of sequential deterministic $\VA$, that is, we show that $\COUNT[\fVA]$ is a hard counting problem, where $\fVA$ is the class of functional $\VA$ (that are not necessarily extended). 
First, we note that $\COUNT[\fVA]$ is not a $\#\textsc{P}$-hard problem -- a property that most of the hard counting problems usually have in the literature~\cite{valiant1979complexity}. 
We instead show that $\COUNT[\fVA]$ is complete for the class $\textsc{SpanL}$~\cite{LOGC}, a counting complexity class that is included in $\#\textsc{P}$ and is incomparable with $\textsc{FP}$, the class of functions computable in polynomial time. 

Intuitively, $\textsc{SpanL}$ is the class of all functions $f$ for which we can find a non-deterministic Turing machine $M$ with an output tape, such that $f(x)$ equals the number of different outputs (i.e. without repetitions) that $M$ produces in its accepting runs on an input $x$, and $M$ runs in logarithmic space. We say that a function $f$ is $\textsc{SpanL}$-complete if $f \in \textsc{SpanL}$ and every function in $\textsc{SpanL}$ can be reduced to $f$ by log-space parsimonious reductions (see~\cite{LOGC} for details).
It is known~\cite{LOGC} that $\textsc{SpanL}$ functions can be computed in polynomial time if, and only if, all the polynomial hierarchy is included in \textsc{P} (in particular $\textsc{NP} = \textsc{P}$).
By well-accepted complexity assumptions the $\textsc{SpanL}$-hardness
of $\COUNT[\fVA]$ hence implies that counting the number of outputs of a $\fVA$ over a document cannot be done in polynomial time.

%\domagoj{Maybe remove the following paragraph?}\martin{I would not remove it, but instead define the class SpanL intuitively.}
%Formally, let $M$ be a non-deterministic Turing machine with output tape where each accepting run of $M$ over an input produces an output. Given an input~$x$, we define $span_M(x)$ as the number of different outputs when running $M$ on $x$. Then, $\textsc{SpanL}$ is the counting class of all functions $f$ for which there exists a non-deterministic logarithmic-space Turing machine with output such that $f(x) = span_M(x)$ for every input $x$. We say that a function $f$ is $\textsc{SpanL}$-complete if $f \in \textsc{SpanL}$ and every function in $\textsc{SpanL}$ can be reduced into $f$ by log-space parsimonious reductions (see~\cite{LOGC} for more details).

\begin{theorem}\label{theo:spanl}
	$\COUNT[\fVA]$ is $\textsc{SpanL}$-complete.
\end{theorem}

It is not hard to see that any functional $\VA$ can be converted in polynomial time into an functional extended $\VA$ (see \cite{MorcianoUV16}). Therefore, the above theorem also implies intractability in counting the number of output mappings of a functional extended $\VA$. 
Given that all other classes of regular spanners studied in this paper (i.e. sequential, non-sequential, etc) include either the class of functional $\VA$ or functional extended $\VA$, this implies that $\COUNT[\LL]$ is intractable for every $\LL$ different from $\LL_1$, the class of deterministic $\sEVA$.

In Section \ref{sec:spanners} we have shown that enumerating the answers of a functional $\VA$ with constant delay can be done after a pre-computation phase that takes the time linear in the document but exponential in the document spanner. The big question that is left to answer is whether enumerating the answers of a functional $\VA$ can be done with a lower pre-computing time, ideally  $O(|\cA|\times |d|)$. Given that constant delay algorithms with efficient pre-computation phases usually imply the existence of efficient counting algorithms~\cite{Segoufin13}, Theorem~\ref{theo:spanl} sheds some light that it may be impossible to find a constant delay algorithm that has pre-computation time better than $O(2^{|\cA|} \times |d|)$, that is obtained by determinizing a  $\fVA$ and running the algorithm from Section \ref{sec:algorithm}. Of course, this does not establish that a constant delay algorithm with precomputation phase sub-exponential in $\cA$ (i.e. $o(2^{|\cA|} \times |d|)$) for $\fVA$ cannot exist, since we are relying on the conjuncture that constant delay algorithms with efficient precomputation phase implies efficient counting algorithms. We leave it as an open problem whether this is indeed true.

%% file: conclusions.tex
We believe that the algorithm described in Section \ref{sec:algorithm}
is a good candidate algorithm to evaluate regular document spanners in
practice.  Throughout the paper we have provided a plethora of
evidence for this claim. First, the proposed algorithm is intuitive
and can be described in a few lines of code, lending itself to easy
implementations. Second, its running time is very efficient for the
class of deterministic sequential extended $\VA$, and the latter in fact
subsumes the class of all regular spanners. Third, we have shown the
cost of executing our algorithm on arbitrary regular spanners,
obtaining bounds that, although not ideal, are reasonable for a wide
range of spanners usually encountered in practice. Finally, we have
shown that better pre-computation times for arbitrary regular spanners
are not very likely, as one would expect to be able to 
compute the number of their outputs more efficiently.

In terms of future directions, we are working on implementing  the algorithm from Section \ref{sec:algorithm} and testing it in practice. We are also looking into the fine points of optimizing its performance, especially with respect to the different translations given in Section \ref{sec:spanners}. As far as theoretical aspects of this work are concerned, we are also interested in establishing hard lower bounds for constant delay algorithms, that do not relay on conjectured claims.

%\stijn{I believe that somewhere we need to mention the following.}
%
%The enumeration problem for relational conjunctive queries reduces to
%the enumeration problem for core spanners (i.e., spanners with
%equality on attributes). Therefor, the latter problem is at least as
%hard as the former, and one can expect that structural properties on
%the core spanner expression will be required to allow for linear time
%pre-computation with constant-delay enumeration (namely: acyclicity
%and free-connexity).
%\cristian{This discussion must be related with the footnote at Section 2.}

% %%% Local Variables: 
% %%% mode: latex
% %%% TeX-master: "../main/main"
% %%% End: 

%% file: appendix-algorithm.tex
%\subsection*{Constant delay version of Algorithm \ref{alg:output}}
%
%\input{../sections/appendix/constantdelayversion}

\subsection*{Proof of Theorem \ref{theo:equivalence}}\label{proof:equivalence}

\input{equivalence}

\subsection*{Proof of Proposition \ref{prop:determinization}}

\input{determinization}

\subsection*{Proof of Lemma \ref{lemma:correctness1}}

\input{correctness1}

%% file: equivalence.tex
%\begin{theorem}\label{theo:equivalence}
%For every variable-set automaton $\cA$ there exists an extended variable-set automaton $\cA'$ such that $\cA \equiv \cA'$ and vice versa. Furthermore, if $\cA$ is sequential, then $\cA'$ is also sequential.
%\end{theorem}

We will show that given a VA $\cA$, one can construct an equivalent extended $\VA$ ($\EVA$) $\cA'$ and vice versa. Both of these constructions have the property that, if the input automaton is sequential or functional, then the output automaton preserves this property.

Let $\cA = (Q, q_0, F, \delta)$ be a VA. The resulting EVA $\cA'$ should produce valid runs that alternate between letter transitions and extended variable transitions. 
To this end, we say that a variable-path between two states $p$ and $q$ is a sequence $\pi: p = p_0  \trans{v_1} p_1 \trans{v_2} \ldots \trans{v_n} p_n = q$ such that $(p_i, v_{i+1}, p_{i+1}) \in \delta$ are variable transitions and $v_i \neq v_j$ for every $i \neq j$. Since all transitions in $\pi$ are variable transitions, we define $\markers(\pi) = \{v_1, \ldots, v_n\}$ as the union of all variable markers appearing in $\pi$.
	
Consider now the following extended VA $\cA'=(Q, q_0, F, \delta')$ where $\delta' = \{ (p,a,q) \in \delta \mid a\in\Sigma \} \cup \delta_{\text{ext}}$ and $(p,S,q) \in \delta_{\text{ext}}$ if, and only if, there exists a variable-path $\pi$ between $p$ and $q$ such that $\markers(\pi) = S$.
Intuitively, this construction condenses variable transitions into a single extended transition. It does so in a way that it can be assured that two consecutive extended transitions are not needed, but also, preserving all possible valid runs. The equivalence $\semd{\cA} = \semd{\cA'}$ for every document $d$ follows directly from the construction and definition of a variable-path.

The opposite direction follows a similar idea, namely, a run in $\cA'$ can be separated into single variable marker transitions in $\cA$ since each extended transition can be separated into a variable-path in $\cA$.
Formally, consider a EVA $\cA' = (Q',q_0', F',\delta')$. The equivalent VA $\cA$ construction is straightforward: for every extended transition between two states, a single path must be created between those two states such that they have the same effect as the single extended transition. The only issue to consider is that one must preserve an order between variable markers in such a way that $\cA$ does not open and close a variable in the wrong order.
To this end, given an arbitrary order $\preceq$ of variables $\VV$, we can expand this order over $\markers_{\VV}$ such that $\Open{x} \preceq \Close{y}$, and $x \preceq y$ implies $\Open{x} \preceq \Open{y}$ and $\Close{x} \preceq \Close{y}$. Namely, two different variable markers follow the original order but all opening markers precede closing markers.
%	
%	\[ x\vdash \preceq y\vdash \preceq z\vdash \preceq \ldots \preceq \dashv x \preceq \dashv y \preceq \dashv z \]
From this, in every extended transition set $S$ we can find a first and last marker in the set, following the mentioned order. Also, we can find for each marker, a successor marker in $S$, as the one that goes after, following the induced order.

Consider now the VA $\cA = (Q' \cup Q_{\text{ext}}, q_0, F, \delta)$ where $Q_{\text{ext}} = \{ \ q_{(p, S, p')}^{v} \mid  (p,S,p') \in \delta' \text{ and } v \in S\}$, $\delta = \{ (p,a,q) \in \delta' \mid a\in\Sigma\} \cup \delta_{\text{first}} \cup \delta_{\text{succ}} \cup \delta_{\text{last}} \cup \delta_{\text{one}}$ and:
\begin{align*}
\delta_{\text{first}} &= \{ (p, v, q_{(p,S, p')}^{v}) \mid  \text{$v$ is the  $\preceq$-minimum element in $S$ } \} \\
\delta_{\text{suc}}&= \{ (q_{(p, S,  p')}^{v}, v', q_{(p, S,  p')}^{v'}) \mid v, v' \in S \text{ and $v'$ is the  $\preceq$-successor of $v$ in $S$} \} \\
\delta_{\text{last}} &= \{ (q_{(p, S, p')}^{v},v' , p') \mid v,v' \in S, \text{ $v'$ is the  $\preceq$-successor of $v$ in $S$, and $v'$ is the $\preceq$-maximum of $S$} \} \\
\delta_{\text{one}} &= \{ (p, v, p') \mid (p, \{v'\}, p')\in \delta'\}
\end{align*}
The previous construction maintains the shape of $\cA'$ but adds the needed intermediate states to form a whole extended marker transition. For every extended transition $(p, S, p')$, $|S|-1$ states are added, labeled with the incoming marker that will arrive to that state. $\delta_{\text{first}}$ defines how to get to the first state of the path, using the first marker of $S$, $\delta_{\text{succ}}$ defines how to get to the next marker in $S$ and $\delta_{\text{last}}$ how to get back to the $\cA'$ state $p'$, having finished the extended transition. $\delta_{\text{one}}$ defines the case when $|S|=1$ and no intermediate states are needed and just use the only marker to do the transition. 
Note that a different set of intermediate states are added for each extended transition $(p,S,p')$, so states do not get reused or transitions do not get mixed. As each transition $(p,S,p')$ of $\cA'$ has a corresponding variable-path in $\cA$, it is obvious that a run in either $\cA$ or $\cA'$ has a corresponding run in the opposite automaton with the same properties, thanks to the order preservation established in the created variable-paths. Finally, it is straightforward to show that $\semd{\cA} = \semd{\cA'}$ for every document $d$.

Let us show that for both constructions, if the input automaton is sequential or functional, then the output automaton preserves this property. In the first case, if $\cA$ is sequential, it is easy to see that all accepting runs of $\cA'$ must be valid, since all extended marker transitions are performed in the same order as in the original automaton $\cA$, and therefore, are also valid. If $\cA$ uses all the variables for all accepting runs, this must also hold for $\cA'$, preserving functionality. %For the translation from $\EVA$ to $\VA$, if $\cA'$ is sequential, then the corresponding runs originated in $\cA$ must also be valid thanks for the order preserving construction over markers. Every run in $\cA'$ is valid, and extending this run using single marker transitions in $\cA$ in such a way that all opening markers are performed before closing ones, then $\cA$ is also sequential. The same occurs with functionality as the case that all variables are used in this run extensions.

%% file: determinization.tex
%\begin{proposition}\label{prop:determinization}
%For every extended variable-set automaton $\cA$ there exists a deterministic extended variable-set automaton $\cA'$  such that $\cA \equiv \cA'$. Furthermore, the size of $\cA'$ is at most exponential in the size of $\cA$.
%\end{proposition}

This result follows from the classical NFA determinization construction. In this case, let $\cA = (Q,q_0, F, \delta)$ be an \EVA, then the following is an equivalent deterministic $\EVA$ for $\cA$: $\cA' = (2^Q ,\{q_0\} , F' , \delta' )$, where $F' = \{ B \in 2^Q \ | \ B \cap F \neq \emptyset \}$ and $\delta'(B, o) = \{ q \in Q \ | \ \exists p \in B. \ (p,o,q) \in \delta \}$. 
%As the NFA determinization, this construction always produces an exponentially big automaton, but, by pruning the automaton by keeping all reachable states from $\{q_0\}$, usually, the resulting automaton is much smaller. 
%The construction may be implemented dynamically by adding states as they are needed, so that not all subsets of $Q$ are always created. 
One can easily check that $\delta'$ is a function and therefore $\cA'$ is deterministic. The fact that $ \semd {\cA} \equiv \semd{\cA'}$ for every document $d$ follows, as well, from NFA determinization: namely, a valid and accepting run in $\cA$ can be translated using the same transitions onto a valid and accepting run in $\cA'$ where the set-states hold the states from the original run. On the other hand, a valid and accepting run in $\cA'$ can only exists if there exists a sequence of states using the same transitions in the original automaton $\cA$.

Finally, since the construction works with sets of $n$ state, then in the worst case it may use $2^n$ states. As for transitions, if each state has all $m$ transitions defined, then the determinization, at most, has $2^n\cdot m$ transitions. 

%% file: correctness1.tex
%\begin{lemma}\label{lemma:correctness1}
%	Let $d = a_1 \ldots a_n$ be a document and $A$ an extended variable-set automaton $A$ that is functional and deterministic. Then for every $0\leq i \leq n$, the following two statements are equivalent:
%	\begin{enumerate}
%		\item there exists a run $\rho \ = \ q_0 \ \trans{S_0} \ p_0 \ \trans{a_1} \ q_1 \ \ldots \ \ \trans{a_i} \ q_{i} \ \trans{S_i} p_i$ over $a_1 \ldots a_i$.
%		\item after executing $\textsc{Capturing}(i)$ in Algorithm \ref{alg:ma-eval}, it holds that $\plist_{p_i} \neq \epsilon$ and \\
%		there is partial output $M$ of \textsc{EnumAll}($\plist_{p_i},\epsilon$) with $M = \Out(\rho)$.
%	\end{enumerate}
%\end{lemma}

	The proof is done by induction on $i$. For the sake of simplification, we denote every object in the $i$-th iteration, that is, while running $\textsc{Reading}(i)$, or $\textsc{Capturing}(i)$, with a superscript $i$. For example, the value of the $\plist_q$ in the $i$-th iteration is denoted by $\plist_q^i$. 

	For the base case assume that $i=0$. At the beginning we have that $\plist_{q_0}=\bot$. If it holds that $\delta(q_0,S)=p_0$ for some $S\neq \emptyset$, then while running $\textsc{Capturing}(1)$, the algorithm will create a new node $n$ with $\texttt{$n$.content}=(S,1)$ and a $\texttt{$n$.list}=\plist_{q_0}^1 = \bot$, and add it at the beginning of the list $\plist_{p_0}^1$. Note that $\plist_{p_0}^1$ can also contain other elements coming from other transitions of the form $\delta(q_0,S')=p_0$ with $S' \neq S$.
	Running then \textsc{EnumAll}($\plist_{p_0}^1,\varepsilon$), will eventually reach the node $n$ in $\plist_{p_0}^1$, resulting in the output $M = (S,1)$.
	Since $n.{\tt list} = \bot$, then the output will be $(S, 1)$ which is the output of the corresponding run $\rho = q_0 \ \trans{S} \ p_0$ and clearly $\Out(\rho)=M$. The other direction is analogous.

	For the inductive step, assume that the claim holds for some $0\leq i < n$. To show that the claim holds for $i+1$ assume first that there is a run:
	\begin{equation}
	\label{run2} \rho_{i+1} \ = \ q_0 \ \trans{S_1} \ p_0 \ \trans{a_1} \ \ldots \ \ \trans{a_i} \ q_i \ \trans{S_{i+1}} \ p_{i} \ \trans{a_{i+1}} \ q_{i+1} \ \trans{S_{i+2}} \ p_{i+1}
	\end{equation}
	that defines an output $\Out({\rho_{i+1}})$. By the induction hypothesis, we know that after running $\textsc{Capturing}(i+1)$ we have that $\plist_{p_i}^{i+1} \neq \epsilon$, and that running \textsc{EnumAll}($\plist_{p_i}^{i+1},\epsilon$) results in an output $M_i$ with ${M_i}=\Out({\rho_i})$, where $\rho_i \ = \ q_0 \ \trans{S_1} \ p_0 \ \trans{a_1} \ q_1 \ \trans{S_2} \ p_1 \ \trans{a_2} \ \ldots \ \ \trans{a_i} \ q_i \ \trans{S_{i+1}} \ p_{i}$. The algorithm now proceeds by executing the procedures $\textsc{Reading}(i+1)$ and $\textsc{Capturing}(i+2)$ one after the other.
	
	Consider what happens when the procedure $\textsc{Reading}(i+1)$ is executed. First, we know that the list $\plist_{p_i}^{i+1}$ gets copied to $\plistold_{p_i}$ and $\plist_{p_i}^{i+1}$ is reset to the empty list $\epsilon$. Then, since $\plistold_{p_i} \neq \epsilon$, and since $q_{i+1}=\delta(p_i,a_{i+1})$, the procedure $\textsc{Reading}(i+1)$ will append the entire list $\plist_{p_i}^{i+1}$ somewhere in the list $\plist_{q_{i+1}}^{i+1}$. Therefore, we know that after executing $\textsc{Reading}(i+1)$, the entire list $\plist_{p_i}^{i+1}$ will appear in the list $\plist_{q_{i+1}}^{i+1}$ before the procedure $\textsc{Capturing}(i+2)$ is executed. 
	
	In $\textsc{Capturing}(i+2)$ we will first guard a copy of 
	$\plist_{q_{i+1}}^{i+1}$ in $\plistold_{q_{i+1}}$. 
	What follows depends on whether $S_{i+2}=\emptyset$ or not. In the case that  $S_{i+2}=\emptyset$, we know that $p_{i+1}=q_{i+1}$ and that the nodes in $\plistold_{q_{i+1}}$ remain on the list $\plist_{p_{i+1}}^{i+2} = \plist_{q_{i+1}}^{i+1}$. The latter follows since any other transition such that $\delta(q,S)=p_{i+1}$ will simply add a new node at the beginning of $\plist_{p_{i+1}}^{i+2}$. Because of this we also have $\plist_{p_{i+1}}^{i+2} \neq \epsilon$. Therefore, after $\textsc{Capturing}(i+2)$ has executed, running \textsc{EnumAll}($\plist_{p_{i+1}}^{i+2}, \varepsilon$) will have $M_i$ with ${M_i}=\Out({\rho_i})$ as one of its outputs, since it will traverse the part of the list $\plist_{q_{i+1}}^{i+1}$ which was already present after $\textsc{Capturing}(i+1)$. Since $\Out({\rho_{i+1}})=\Out({\rho_i})$, the result follows.

	On the other hand, if $S_{i+2}\neq \emptyset$, since $\delta(q_{i+1},S_{i+2})=p_{i+1}$, and $\plistold_{q_{i+1}} \neq \epsilon$, the procedure 
	$\textsc{Capturing}(i+2)$ will create a new node $n$ to be added to the list $\plist_{p_{i+1}}^{i+2}$. The node $n$ will have the values $n.{\tt content} = (S_{i+2}, i+2)$ and $n.{\tt list} = \plistold_{q_{i+1}} = \plist_{q_{i+1}}^{i+1}$. In particular, after $\textsc{Capturing}(i+2)$, we have that $\plist_{p_{i+1}}^{i+2}\neq \epsilon$ and that running \textsc{EnumAll}($\plist_{p_{i+1}}^{i+2},\varepsilon$) will eventually do a call to the procedure \textsc{EnumAll}($n.{\tt list}$, $(S_{i+2},i+2) \cdot \varepsilon$). Therefore one of the outputs of the original call \textsc{EnumAll}($\plist_{p_{i+1}}^{i+2},\varepsilon$) will simply append the pair $(S_{i+2},i+2)$ to $M_i$ resulting in $M_{i+1}= M_i \cdot (S_{i+2},i+2)$ as output. From (\ref{run2}) it is clear that $M_{i+1}=\Out(\rho_{i+1})$.
	
	For the other direction, assume now that we have executed the procedure $\textsc{Capturing}(i+2)$ in Algorithm \ref{alg:ma-eval} and that $\plist_{p_{i+1}}^{i+2} \neq \epsilon$. Furthermore, assume that $n\neq \bot$ is a node in $\plist_{p_{i+1}}^{i+2}$ and $(S, j) = n.{\tt content}$.
	The first observation we make is that for any node $n' \neq \bot$ that is inside the list $n.{\tt list}$ with $(S', j') = n'.{\tt content}$, it holds that $j' < j$. This is evident from the algorithm since the only way that the node $n'$ can enter the list $n.{\tt list}$ is when the node $n$ is being created in $\textsc{Capturing}(j)$. However, in this case, the node $n'$ must have already been defined in some previous iteration of the algorithm (as $n.{\tt list}$ guards the ``old'' pointers $\plistold_p$ for some $p$), and since new nodes are being created only in the procedure \textsc{Capturing}, this means that $n'$ was created in $\textsc{Capturing}(j')$. Because of this we have that $j' < j$. Moreover, given that each iterative call of \textsc{EnumAll} uses elements from the list $n.{\tt list}$ we have that for any output $M = (S_1,i_1) \ldots (S_k,i_k)$ of \textsc{EnumAll}($\plist_{p_{i+1}}^{i+1}, \epsilon$), it holds that $i_k>i_{k-1}>\cdots >i_1$.
	
	Let $M_{i+1} = (S_1,i_1) \ldots (S_{k-1},i_{k-1}) (S_k,i_k)$, where $k\geq 0$, be one output of \textsc{EnumAll}($\plist_{p_{i+1}}^{i+2}, \varepsilon$). There are two possible cases: either $i_k=i+2$, or $i_k\neq i+2$. Consider first the case when $i_k\neq i+2$. In this case, the procedure \textsc{EnumAll}($\plist_{p_{i+1}}^{i+2},\varepsilon$) will not access a node created in $\textsc{Capturing}(i+2)$ when generating the output $M_{i+1}$. Therefore, it will have to start with some node $n$ that got in the list $\plist_{p_{i+1}}^{i+2}$ during $\textsc{Reading}(i+1)$. This can only happen if $\delta(p_{i},a_{i+1})=q_{i+1} = p_{i+1}$, for some state $p_i$ such that $\plistold_{p_{i}}\neq \epsilon$ at the beginning of $\textsc{Reading}(i+1)$, and $n$ belongs to $\plist_{p_i}^{i+1}$, since the only thing $\textsc{Reading}(i+1)$ does is to copy and merge the lists $\plist_{p_i}^{i+1}$, for different states $p$. However, this means that $\plist_{p_i}^{i+1} \neq \epsilon$ after executing $\textsc{Capturing}(i+1)$. 
	This means that $M_{i+1}$ is one of the outputs of \textsc{EnumAll}($\plist_{p_i}^{i+1}, \varepsilon$) after executing $\textsc{Capturing}(i+1)$. Using the induction hypothesis, there is  a run $\rho_i \ = \ q_0 \ \trans{S_1} \ p_0 \ \trans{a_1} \ q_1 \ \ldots \  \trans{a_i} \ q_i \ \trans{S_{i+1}} \ p_{i}$ such that $\Out({\rho_i})={M_{i+1}}$. Consider now the run $\rho_{i+1} \ = \ q_0 \ \trans{S_1} \ p_0 \ \trans{a_1} \ q_1 \ \ldots \ \ \trans{a_i} \ q_i \ \trans{S_{i+1}} \ p_{i} \ \trans{a_{i+1}} \ q_{i+1} \ \trans{\emptyset} p_{i+1}$. Since clearly $\Out({\rho_{i+1}})=\Out({\rho_{i}})$, we get that the claim holds true for $i+1$ when $i_k\neq i+2$.
	
	Consider now the case when $i_k= i+2$. To produce $M_{i+1}$ as output, the procedure \textsc{EnumAll}($\plist_{p_{i+1}}^{i+2},\varepsilon$) had to do a recursive call to \textsc{EnumAll}($n.{\tt list}, (S_k,i+2) \cdot \varepsilon$), for some node $n$ on $\plist_{p_{i+1}}^{i+2}$. Since $i_k= i+2$ we know that node $n$ was created in $\textsc{Capturing}(i+2)$. Therefore, there must exist a state $q_{i+1}$ such that $\delta(q_{i+1},S_{k})=p_{i+1}$ and $\plistold_{q_{i+1}} \neq \epsilon$ at the beginning of $\textsc{Capturing}(i+2)$. 
	As $n.{\tt list} = \plistold_{q_{i+1}}$, we know that running \textsc{EnumAll}($\plistold_{q_{i+1}},\varepsilon$) must have $M_i=(S_1,i_1) \ldots (S_{k-1},i_{k-1})$ as one of its outputs. However, since all the nodes in $\plistold_{q_{i+1}}$ must already be in $\plist_{q_{i+1}}^{i+1}$ after running $\textsc{Reading}(i+1)$, they must enter this list in $\textsc{Reading}(i+1)$ because there is some transition $\delta(p_i,a_{i+1})=q_{i+1}$, for some state $p_i \in Q$. In particular, the recursive call of \textsc{EnumAll}($n.{\tt list}, \varepsilon$) that produced $M_i$ as its output used a node on $\plist_{q_{i+1}}^{i+1}$ that was already present in $\plist_{p_i}^{i+1}$, for state $p_i$ such that $\delta(p_i,a_{i+1})=q_{i+1}$. By the induction hypothesis, there is a run $\rho_i = q_0 \ \trans{S_1} \ p_0 \ \trans{a_1} \ q_1  \ \ldots \ \trans{a_i} \ q_i \ \trans{S_{i+1}} \ p_{i}$ such that $\Out({\rho_i})={M_{i}}$. Because of this, the run $\rho_{i+1} = q_0 \ \trans{S_1} \ p_0 \ \trans{a_1} \ q_1 \ \ldots \ \ \trans{a_i} \ q_i \ \trans{S_{i+1}} p_{i} \ \trans{a_{i+1}} \ q_{i+1} \ \trans{S_k} \ p_{i+1}$ clearly has $\Out({\rho_{i+1}})=M_{i+1}$. This concludes the proof.

%% file: appendix-spanners.tex
\subsection*{Proof of Proposition \ref{prop:general-VA}}

\input{general_va}

\subsection*{Proof of Proposition \ref{prop:seq-extended-existance}}

\input{seq_ext_ex}	

\subsection*{Proof of Proposition \ref{prop:func-extended-blowup}}

\input{func_ext}	
	
%\subsection*{Proof of Proposition \ref{prop:seq-extended-blowup}}
%
%\input{../sections/appendix/seq_ext}

\subsection*{Proof of Proposition \ref{prop:algebra-eVA}}

\input{algebra}

\subsection*{Proof of Proposition \ref{prop:reg_span_compiler}}

\input{expressions-to-va}	

\subsection*{Proof of Proposition \ref{prop:func_algebra_compilation}}

\input{func_algebra_compilation}

%% file: general_va.tex
%!TEX root = ../../main/main.tex

Let $\cA=(Q,q_0,F,\delta)$ be a $\VA$ with $|Q|=n$, $|\delta|=m$ and $\ell$ variables. We show how to construct a $\dsEVA$ $\cA'=(Q',q_0',F',\delta')$ that is equivalent to $\cA$ and has $2^n\times 3^\ell$ states, $2^n3^\ell(2^\ell+|\Sigma|)$ transitions and $\ell$ variables. Let us first describe the set $Q'$ of states of $\cA$. Intuitively, every state will correspond to a tuple $(\{q_1,\ldots,q_k\},S)$, where $q_1,\ldots,q_k\in Q$ are the states reached by $\cA$ by reading the set of variable markers $S$. Since there are $n$ states, the first component (the set of reached states) can be chosen of $2^n$ different sets. Now for each state in the chosen set, we have a set of variable markers. Note that we need to exclude the sets of variable markers that contain a variable that is closed but not opened. Therefore, if we have $\ell$ variables the number of such sets of variable markers is $\sum_{i=1}^\ell\binom{n}{i}2^i$, where $i$ represents the number of opened variables, $\binom{n}{i}$ the different ways of choosing those $i$ variables, and $2^i$ possible ways of closing those variables. From this we obtain

$$\sum_{i=0}^\ell\binom{n}{i}2^i=\sum_{i=0}^\ell\binom{n}{i}2^i1^{n-i}=(1+2)^\ell=3^\ell$$

Therefore, it is clear that we have $2^n3^\ell$ states. The only initial state is $q_0'=(\{q_0\},\emptyset)$.

Let us now define the set of transitions $\delta'$. Given a character $c\in\Sigma$, the transition $\delta((P,S),c)$ is simply defined as $(\delta(P,c),S)$, where $\delta(P,c)=\{q\in Q\mid\exists q'\in P \text{ s.t. } (q',c,q)\in\delta\}$. Let us now describe the variable transitions. Intuitively, $\delta'((P,S), S')$ will contain the set of states that can be reached from a state of $P$ by following a path in $\cA$ of variable transitions in which each variable marker in $S'$ is mentioned exactly once. Formally, we define a variable path in $\cA$ as a sequence of transitions $p=(q_{i_1},m_1,q_{i_2})(q_{i_2},m_2,q_{i_3})\ldots(q_{i_{h-1}},m_{h-1},q_{i_h})$ in $\delta^*$, where each $m_j$ is a variable marker and for all $j\neq k\in\{0,\ldots,m\}$ we have $m_j\neq m_k$. If $S=\{m_1,\ldots,m_h\}$ we say that $p$ is an $S$-path from $q_{i_1}$ to $q_{i_h}$. Then, for every $P\subset Q$ and every pair $(S,S')$ of variable markers such that $S$ and $S'$ are compatible (in the sense that every closed variable in $S'\cup S$ is also opened), $\delta'((P,S), S')$ is defined as $(P',S'')$ where:
\begin{enumerate}
	\item $S''=S\cup S'$ and
	\item for every $q'\in P'$ there exists $q\in P$ such that there is an $S'$-path between $q$ and $q'$.
\end{enumerate}
If $S$ and $S'$ are not compatible, $\delta'((\{q_0,\ldots,q_k\},S), S')$ is undefined (note that this makes the automaton sequential).

Let us analyze the number of transitions in $\delta'$. To do a fine-grained analysis of the variable transitions, for each $i\in\{0,\ldots,\ell\}$ we consider the number of states $(P,S)$ where $S$ has exactly $i$ open variables (i.e. $2^n\binom{\ell}{i}2^i$), multiplied by the number of variable transitions that can originate in such a state. This number is again analyzed for the $\binom{\ell-i}{j}$ sets of size $j$ of opened variables (out of the $\ell-i$ remaining variables), and for each of these sets which variables are closed ($2^j$ possibilities). The variable transitions are

$$\sum_{i=0}^\ell\left[2^n\binom{\ell}{i}2^i\sum_{j=0}^{\ell-i}\binom{\ell-i}{j}2^{j}\right]\ =\ 2^n\sum_{i=0}^\ell\binom{\ell}{i}2^i3^{\ell-i}\ =\ 2^n(2+3)^\ell\ =\ 2^n5^\ell$$
This is the number of variable transitions in $\cA'$. To this number, we must add the number of character transitions, which is at most one transition per state per character, i.e. $2^n3^\ell|\Sigma|$. Then, the total number of transitions is $2^n5^\ell+2^n3^\ell|\Sigma|$ as expected. Finally, we define the set $F'$ of final states as those states $(P,S)$ in which $P\cap F\neq\emptyset$ and all variables opened in $S$ are also closed.

It is trivial to see that $\cA'$ is sequential. Since the only way to reach a state $(P,S)$ using a variable transition is from a previous state $(P',S')$ and a set of markers $S''$ such that $S'\cup S''=S$, it is clear that if a run $\rho$ ends in state $(P,S)$ then $S$ is the union of all variable markers seen in $\rho$. Sequentiality then follows since we require at all times that every variable is opened and closed at most once, variables are opened before they are closed, and in final states all opened variables are closed. The fact that $\cA'$ is deterministic can be immediately seen from the construction; for every state there is at most one transition for each character, and at most one transition for each set of variable markers. Since $\cA'$ is an extended $\VA$ and must alternate between variable and character transitions, this implies that two different runs cannot generate the same mapping.

We now show that $\cA$ is equivalent to $\cA'$. Let $d$ be a document and assume the mapping $\mu$ is produced by a valid accepting run $\rho \ = \ (q_0, i_0) \ \trans{o_1} \ (q_1, i_1) \ \trans{o_2} \ \cdots \ \trans{o_m} \ (q_m, i_m)$ of $\cA$ over $d$. Define a function $f$ with domain $i\in\{1,\ldots,m\}$  as follows:
$$f(i)=\left\{\begin{array}{ll}k & \text{ if $\forall j\in\{1,\ldots,k\}$ $o_i$ is a variable marker, and either $k=m$ or $o_{k+1}$ is not a variable marker.}\\(o_i,\emptyset) & \text{ if $o_i\in\Sigma$ and $o_{i+1}\in\Sigma$}\\o_i & \text{ otherwise.}\end{array} \right.$$
With this definition, we construct a run for $\cA'$ generating $\mu$ starting with $\rho'$ as the run that only contains $q_0'$ and $i=1$ as follows:
\begin{enumerate}
	\item If $f(i)=k$, define the set of variable markers $S$ as $\bigcup_{j=i}^k o_j$, update $\rho'$ to $\rho'\trans{S}\delta'((P',S'), S)$, where $(P',S')$ was the last state of $\rho'$ before this update. Finally, update $i$ to $k+1$.
	\item If $f(i)=(o_i,\emptyset)$, update $\rho'$ to $\rho'\trans{o_i}\delta'(P',S')\trans{\emptyset}\delta'((P',S'),\emptyset)$, where $(P',S')$ was the previous final state of $\rho'$. Finally update $i$ to $i+1$.
	\item If $f(i)=o_i$, update $\rho'$ to $\rho'\trans{o_i}\delta'((P',S'),o_i)$, where $(P',S')$ was the previous final state of $\rho'$. Finally update $i$ to $i+1$.
\end{enumerate}
We need to show that this is actually a valid and accepting run of $\cA'$ over $d$. To show that it is a run over $\cA'$ is simple: since $\rho$ is a run over $\cA$, the construction of $f$ implies that for every step of the form $(P',S')\trans{S}(P,S\cup S')$ in $\rho'$ there is an $S$-path from a state in $P$ to a state in $P'$ (assuming $S\neq\emptyset$). The $\emptyset$ and character transitions immediately yield valid transitions for $\rho'$. The fact that $\rho'$ follows from the construction, as we can see that it will open and close variables in the same order and in the same positions as $\rho$, which was already valid. This also shows that $\rho'$ generates $\mu$. The fact that $\rho'$ is valid follows because $q_m\in F$ is final and belongs to the last state of $\rho'$.\par

The opposite direction is similar: considering a mapping $\mu$ generated by a valid accepting run $\rho'$ of $\cA'$ over $d$, we need to show a valid accepting run $\rho$ of $\cA$ over $d$ generating $\mu$. We omit this direction as $\rho$ can be generated by doing essentially the same process as before but in reverse: We know that $\rho'$ ends in a state that mentions a final state $q_f\in F$. Then, for each step $(P,S)\trans{o}(P',S')$ of $\rho$ and the selected state in $P'$ (at the beginning, $q_f$), there is a transition or an $(S'\setminus S)$-path going from a state $q\in P$ to $q'$. This way we can construct $\rho'$ backwards; proving it is valid, accepting and it generates $\mu$ follows again by the construction.

%% file: seq_ext_ex.tex
%For every $k>0$ there is a sequential VA $\cA$ with $3\ell+3$ states and $4\ell+1$ transitions, such that for every extended VA $\cA'$ such that $\cA\equiv\cA'$ it is the case that $\cA'$ has at least $2^\ell$ transitions.

Unfortunately, a sequential $\VA$ has an exponential blow-up in terms of the number of transitions the resulting $\EVA$ may have. For every $\ell$ consider the VA $\cA$, with $3\ell+2$ states and $4\ell+1$ transitions depicted in Figure \ref{fig:badsequentialappendix} with $2\ell$ variables: $x_1,\ldots,x_\ell, y_1, \ldots, y_\ell$. $\cA$ only produces valid runs for the document $d = a$, the resulting mapping is always valid but never total, as it properly opens and closes variables, but never all of them. At each intermediate state, the run has the option to choose opening and closing either $x_i$ or $y_i$, for every $1 \leq i \leq \ell$, generating $2^\ell$ different runs. Therefore, if we only consider the equivalent $\EVA$ that extends transitions from $q_0$ to $q$ and no other pair in between, we obtain the extended $\VA$ $\cA'$ in Figure~\ref{fig:badsequentialextendedappendix}. This is the smallest $\EVA$ equivalent to $\cA$, since each of the mentioned transitions group the greatest amount of variables in a different run. Specifically, each of this transitions has a corresponding and different $\epsilon$ mapping, the one where the contained variables is defined. Therefore, it has $2^\ell$ transitions, as well as any other equivalent $\EVA$.

 \begin{figure}
 	\begin{center}
 			\begin{tikzpicture}[->,>=stealth',auto, thick,initial text= {}]
			
 			% The graph
 			\node [state,initial] at (0,0) (q0) {$q_0$};
 			\node [state] at (1.5,1) (q01) {};
 			\node [state] at (1.5,-1) (q02) {};		  
 			\node [state] at (3,0) (q1) {};
 			\node [state] at (4.5,1) (q11) {};
 			\node [state] at (4.5,-1) (q12) {};
 			\node [inner sep=2mm] at (6,0) (q2) {$\ldots$};
 			\node [state] at (7.5,1) (qn1) {};
 			\node [state] at (7.5,-1) (qn2) {};
 			\node [state] at (9,0) (qn) {$q$};
 			\node [state,accepting] at (11,0) (qf) {};

 			%\node [state,accepting] at (10,0) (q5) {$q_5$};
			
 			% Graph edges
 			\path[->] (q0) edge node[above] {$\Open{x_1} \ \ \ $} (q01);
 			\path[->] (q0) edge node[below] {$\Open{y_1} \ \ \ $} (q02);
 			\path[->] (q01) edge node[above] {$\ \ \Close{x_1}$} (q1);
 			\path[->] (q02) edge node[below] {$\ \ \Close{y_1}$} (q1);	  
 			\path[->] (q1) edge node[above] {$\Open{x_2} \ \ \ $} (q11);
 			\path[->] (q1) edge node[below] {$\Open{y_2} \ \ \ $} (q12);
 			\path[->] (q11) edge node[above] {$\ \  \Close{x_2}$} (q2);
 			\path[->] (q12) edge node[below] {$ \ \ \Close{y_2}$} (q2);
 			\path[->] (q2) edge node[above] {$\Open{x_\ell} \ \ \ $} (qn1);
 			\path[->] (q2) edge node[below] {$\Open{y_\ell} \ \ \ $} (qn2);
 			\path[->] (qn1) edge node[above] {$\ \ \Close{x_\ell}$} (qn);
 			\path[->] (qn2) edge node[below] {$\ \ \Close{y_\ell}$} (qn);
 			\path[->] (qn) edge node[above] {$a$} (qf);
 			\end{tikzpicture}
 	\end{center}
 	\vspace*{-15pt}
 	\caption{A sequential $\VA$ with $2\ell$ variables such that every equivalent $\EVA$ has $O(2^\ell)$ transitions.}
 	\label{fig:badsequentialappendix}
 \end{figure}
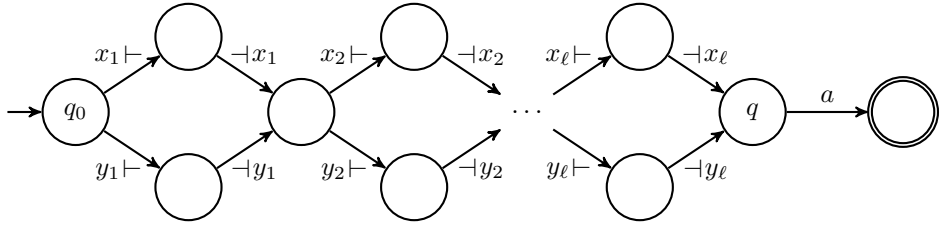

  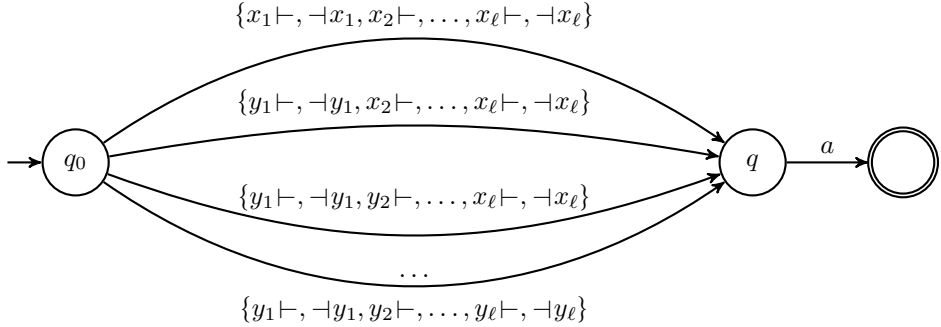
\begin{figure}
 	\begin{center}
 			\begin{tikzpicture}[->,>=stealth',auto, thick, scale = 1.0,initial text= {}]
			
 			% The graph
 			\node [state,initial] at (0,0) (q0) {$q_0$};
 			\node [state] at (9,0) (qn) {$q$};
 			\node [inner sep=2mm] at (4.5,-1.5) (q2) {$\ldots$};
 			\node [state,accepting] at (11,0) (qf) {};

 			%\node [state,accepting] at (10,0) (q5) {$q_5$};
			
 			% Graph edges
 			\path[->] (q0) edge[bend left=35] node[above] {$\{ \Open{x_1}, \Close{x_1}, \Open{x_2},\ldots, \Open{x_\ell}, \Close{x_\ell} \} $} (qn);
 			\path[->] (q0) edge[bend left=10] node[above] {$\{ \Open{y_1}, \Close{y_1}, \Open{x_2},\ldots, \Open{x_\ell}, \Close{x_\ell} \} $} (qn);
 			\path[->] (q0) edge[bend right=20] node[above=0.2] {$\{ \Open{y_1}, \Close{y_1}, \Open{y_2},\ldots, \Open{x_\ell}, \Close{x_\ell} \} $} (qn);
 			\path[->] (q0) edge[bend right=35] node[below] {$\{ \Open{y_1}, \Close{y_1}, \Open{y_2}, \ldots, \Open{y_\ell}, \Close{y_\ell} \} $} (qn);
 			\path[->] (qn) edge node[above] {$a$} (qf);
 			\end{tikzpicture}
 	\end{center}
 	\vspace*{-15pt}
 	\caption{The smallest $\EVA$ $\cA'$ equivalent to $\cA$ with $2^\ell$ transitions.}
 	\label{fig:badsequentialextendedappendix}
 \end{figure}

%% file: func_ext.tex
%Let $\cA$ be a functional $\VA$ with $n$ states, $m$ transitions and $\ell$ variables. There is a $\DESVA$ $\cA'$ with at most $2^n$ states, $2^n(n^2+|\Sigma|)$ transitions and $\ell$ variables such that $\cA\equiv\cA'$.

We showed in the proof of Theorem \ref{theo:equivalence} that given a $\VA$ $\cA$ we can  construct an equivalent $\EVA$ $\cA_\text{ext}$, and the functional property also holds for $\cA_\text{ext}$. We show here that if $\cA$ has $n$ states and $m$ transitions, then $\cA_\text{ext}$ has at most $n$ states and $m + n^2$ transitions. 

The bound $n$ over the number of states in $\cA_\text{ext}$ directly follows from the construction in Theorem~\ref{theo:equivalence}. 
The bound $m + n^2$ over the number of transitions in $\cA_\text{ext}$ follows from the fact that $\cA$ is functional, given that in a functional $\VA$ the number of extended marker transitions that can be established between two states is at most one. Specifically, we prove the following lemma\footnote{A similar lemma appears in \cite{FKP17}.} (for a formal definition of variable path see the Proof of Theorem \ref{theo:equivalence}).
\begin{lemma}
	If $\cA$ is functional, then for every two states $q$ and $q'$ in $\cA$ that can produce valid runs, it holds that $\markers(\pi) = \markers(\pi')$ for every pair of variable paths $\pi$ and $\pi'$ between $q$ and $q'$.
\end{lemma}
\begin{proof}
	If not, then there are two states $q$ and $q'$ in $\cA$ such that the are at least two variable paths $\pi$ and $\pi'$ between $q$ and $q'$, with different sets of markers appearing in them. Since $q$ and $q'$ can produce a valid run, then they are both reachable from $q_0$ and can reach a final state. Specifically, let $\pi_i$ be the path from $q_0$ to $q$, and $\pi_f$ be the path from $q'$ to a final state. Then the concatenated paths $\pi_i \pi \pi_f$ and $\pi_i \pi'\pi_f$ are both accepting. Both also must be valid, because $\cA$ is functional. But, the set of markers in $\pi$ and $\pi'$ are different, yet, the rest of the paths are the same and they open and close all variables in $\cA$. This is a contradiction: either $\pi_i \pi \pi_f$ or $\pi_i \pi'\pi_f$ cannot open and close all variables. Therefore, all paths between $q$ and $q'$ must contain the same set of markers appearing in them. 
\end{proof}
	
Thanks to the previous lemma, we can bound the number of possible extended marker transitions between every pair of states to just one: the set of markers appearing in paths connecting these two states. 
%One can extend the last proved property to the case of $q_0$ and every other reachable state. Since all paths from $q_0$ to an specific state have the same set of markers, then each state has a set of opened and closed markers that all runs share up to that state. 
Therefore, using our construction for $\cA_\text{ext}$, at most one extended marker transition may be added between two states. Then, additionally to the $m$ transitions in $\cA$, at most $n^2$ extended marker transitions can be added (for every pair of states in $\cA$). 
We conclude that $\cA_\text{ext}$ has at most $m + n^2$ transitions.

Finally, as showed in Proposition \ref{prop:determinization}, $\dsEVA$ $\cA'$ can be constructed such that $\cA_\text{ext} \equiv \cA'$, where, $\cA'$ has at most $2^n$ states. Since $\cA'$ is deterministic, every state can have, at most, the number of extended transitions added or all the possible symbols in $\Sigma$. Therefore, the number of transitions for $\cA'$ is at most $2^n(n^2+|\Sigma|)$.

%% file: algebra.tex
\subsubsection*{Join of functional extended \VA}
%For any pair of $\fEVA$ $\cA_1$ and $\cA_2$, there exists an $\EVA$ $\cA'$ of quadratic size such that $\cA'\equiv\cA_1 \Join \cA_2$.

Let $\cA_1 = (Q_1, q_0^1, F_1, \delta_1)$ and $\cA_2 = (Q_2, q_0^2, F_2, \delta_2)$ be two $\fEVA$. Let $\VV_1 = \var(\cA_1)$, $\VV_2 = \var(\cA_2)$ and $\VV_\Join = \VV_1 \cap \VV_2$. The intuition behind the following construction is similar to the standard construction for intersection of NFAs: we run both automaton in parallel, limiting the possibility to use simultaneously markers on both automata only on shared variables, and let free  use of markers that are exclusive to $\VV_1$ or $\VV_2$. Formally, we define $\cA_\Join = (Q_1 \times Q_2 , (q_0^1, q_0^2), F_1 \times F_2, \delta)$ where $\delta$ is defined as follows: 

\begin{itemize}
\item $\big( (p_1,p_2), a, (q_1, q_2))\big)\in\delta$ \ if \ $a \in \Sigma$, $(p_1,a,q_1) \in \delta_1$ and $(p_2,a,q_2) \in \delta_2$.
\item $\big( (p_1,p_2), S_1, (q_1, p_2))\big)\in\delta$  \ if \  $p_2\in Q_2$,  $(p_1,S_1,q_1) \in \delta_1$, and $S_1\cap \markers_{\VV_\Join} = \emptyset$.
\item $\big( (p_1,p_2), S_2, (p_1, q_2))\big)\in\delta$  \ if \  $p_1\in Q_1$, $(p_2,S_2,q_2) \in \delta_2$ and $S_2\cap \markers_{\VV_\Join} = \emptyset$. 
\item $\big( (p_1,p_2), S_1\cup S_2, (q_1, q_2))\big)\in\delta$  \ if \  $(p_1,S_1,q_1) \in \delta_1$, $(p_2,S_2,q_2) \in \delta_2$, and 
$S_1 \cap \markers_{\VV_\Join} = S_2 \cap \markers_{\VV_\Join}$.
\end{itemize}

To show that $\semd{\cA_\Join} \subseteq \semd{\cA_1} \Join \semd{\cA_2}$, let $\mu$ be a mapping in $\semd{\cA_\Join}$ for the document $d$, and $\rho_\mu$ the corresponding valid and accepting run of $\cA_\Join$ over $d$. 
By construction, from $\rho_\mu$ we can get a sequence of states in $\cA_1$ and $\cA_2$ that define runs $\rho_1$ and $\rho_2$ in their respective automaton. This preserves both order and positions of markers. Since $\cA_1$ and $\cA_2$ are functional and $\rho_\mu$ is accepting, then $\rho_1$ and $\rho_2$ are accepting and valid runs of $\cA_1$ and $\cA_2$, respectively. This implies that $\mu^{\rho_1} \in \semd{\cA_1}$ and $\mu^{\rho_2} \in \semd{\cA_2}$. Finally, since all common marker transitions are performed by both automata at the same union transitions, then $\mu^{\rho_1} \sim \mu^{\rho_2}$ and therefore $\mu = \mu^{\rho_1} \cup \mu^{\rho_2} \in \semd{\cA_1} \Join \semd{\cA_2} $. 

To show that $\semd{\cA_1} \Join \semd{\cA_2} \subseteq \semd{\cA_\Join}$, let $\mu_1\in\semd{\cA_1}$, $\mu_2\in\semd{\cA_2}$ such that $\mu_1 \sim \mu_2$ and $\rho^{\mu_1}$ and $\rho^{\mu_2}$ be their corresponding runs. Since they are compatible mappings, then both runs use each marker in $\markers(\VV_\Join)$ in the same positions of $d$. Therefore, by merging the marker transitions made in each run, the corresponding union transitions must exists in $\cA_\Join$ and used to construct a run $\rho$ in $\cA_\Join$. Finally, since $\rho_1$ and $\rho_2$ are accepting, valid, and total, then $\rho$ is also accepting, valid and total for $\var(\cA_1) \cup \var(\cA_2)$, that is, $\mu^{\rho} \in \semd{\cA_\Join}$. It is easy to see that $\mu^{\rho} = \mu_1 \cup \mu_2$, and therefore $\mu_1 \cup \mu_2 \in \semd{\cA_\Join}$.

To show that $\cA_\Join$ is also functional, let $\rho$ be an accepting run in $\cA_\Join$ for $d$. Thanks to the construction, and as shown before, corresponding runs in $\cA_1$ and $\cA_2$ can be produced from $\rho$ that are also accepting, and therefore valid and total since they are functional. Since all common markers are used in the same positions and precisely once in the corresponding runs, this is also true for $\rho$. Also, all variables are used in runs of $\cA_1$ and $\cA_2$, therefore $\rho$ is valid and total for $\var(\cA_1) \cup \var(\cA_2)$.
Regarding the size of $\cA_\Join$, one can verify that $\cA_\Join$ has $|Q_1|\times |Q_2|$ states and at most $O(|\delta_1| \times |\delta_2|)$ transitions. Therefore, $\cA_\Join$ is quadratic in size.

\subsubsection*{Projection of functional extended \VA}
%Let $\cA$ be a $\fEVA$ and $Y\subset \VV$, there exists $\EVA$ $\cA_\pi$ such that $\cA_\pi \equiv \pi_Y \cA_1 $ of linear size.

To prove this, for the sake of simplification we use the notion of $\epsilon$-transitions in $\EVA$, as the usual notion for regular NFA, namely, transition of the form $(q, \epsilon, p)$. 
As it is standard in automata theory, if a run uses an $\epsilon$-transition, this produces no effect on the document read or variables that are opened or closed, and only the current state of the automaton changes from $q$ to $p$.
Furthermore, in the semantics of $\epsilon$-transitions we assume that no two consecutive $\epsilon$-transitions can be used. 
Clearly, $\epsilon$-transitions do not add expressivity to the model and only help to simplify the construction of the projection.

Let $\cA =(Q,q_0, F,\delta )$ be a $\fEVA$ and $Y\subset \VV$. Let $U = \markers_\VV \setminus \markers_Y$ be markers for unprojected variables, then $\cA_\pi = (Q,q_0, F, \delta')$ where $(q,a,p) \in \delta'$ whenever $(q,a,p) \in \delta$ for every $a\in \Sigma$, $(q,S \setminus U,p) \in \delta'$ whenever $(q,S,p) \in \delta$ and $S\setminus U \neq \emptyset$, and $(q,\epsilon,p) \in \delta'$ whenever  $(q,S,p) \in \delta$ and $S \setminus U = \emptyset$.

The equivalence between $\cA$ and $\cA_\pi$ is straightforward. For every $\mu \in \semd{\cA}$, there exists an accepting and valid run $\rho$ in $\cA$ over $d$.  For $\rho$ there exists a run $\rho'$ in $\cA_\pi$ formed by the same sequence of states, but extended marker or $\epsilon$-transitions are used that only contain markers from $Y$. Moreover, $\rho'$ must also be valid since it maintains the order of $Y$-variables used in $\rho$.  This shows that $\pi_Y \semd{\cA} \subseteq \semd{\cA_\pi}$.  The other direction, $\semd{\cA_\pi} \subseteq \pi_Y \semd{\cA}$, follows from the fact that $\cA'$ has no additional accepting paths in comparison to $\cA$. It is also easy to see that $\cA_\pi$ must be functional.

Finally, it is important to note that, as for classical NFAs, from $\cA'$ an equivalent $\epsilon$-transition free $\EVA$ can constructed using $\epsilon$-closure over states. 

\subsubsection*{Union of functional extended \VA}

This construction is the standard disjoint union of automaton,  with $\epsilon$-transitions to each corresponding initial state. Let $\cA_1 = (Q_1, q_0^1, F_1, \delta_1)$ and $\cA_2 = (Q_2, q_0^2, F_2, \delta_2)$ be two $\fEVA$ such $Q_1\cap Q_2 = \emptyset$. Then, $\cA_\cup = (Q_1\cup Q_2, q_0, F_1\cup F_2, \delta_1\cup \delta_2 \cup \{ (q_0, \epsilon, q_0^1), (q_0, \epsilon, q_0^2) \} )$ where $q_0$ is a fresh new state. This simply adds a new initial state connected with $\epsilon$-transitions to the initial states of $\cA_1$ and $\cA_2$, respectively. Therefore every run in $\cA_\cup$ must produce a run from $\cA_1\cup\cA_2$ and vice versa. An equivalent $\epsilon$-transition free automaton can be constructed as in the projection case.

%% file: expressions-to-va.tex
%!TEX root = ../../main/main.tex

Let $\gamma$ be a regular spanner in $\VA^{\{\pi,\cup,\Join\}}$ that uses $k$ functional $\VA$ as input, each of them with at most $n$ states. By Proposition~\ref{prop:algebra-eVA}, we know that we can construct the product between two automata with $n$ and $m$ states, and the resulting automaton will have at most $nm$ states and $nm$ transitions. Moreover, projections and unions remain linear in the size of the input automata. Therefore, if we apply the transformations of Proposition~\ref{prop:algebra-eVA} in a bottom-up fashion to $\gamma$, each algebraic operation will multiply the number of states and transitions of the resulting automaton by $n$. It is trivial to prove then by induction that the final automaton will have at most $n^k$ states and $n^k$ transitions. This automaton needs to be determinized at the end. By Proposition~\ref{prop:func-extended-blowup}, the result will have $2^{n^k}$ states and $n^{2k}+|\Sigma|$ transitions, concluding the proof.

%% file: func_algebra_compilation.tex
\input{union_quadratic}

%% file: union_quadratic.tex
Contrary to the previous proposition, the idea here is to first determinize each automaton and then apply the join and union construction of functional $\EVA$. Given that each automaton will have size $O(2^n)$ after determinization, then the product (e.g. join) of two automata of size $O(2^{n})$ will have size $O(2^{2n})$. Therefore, the number of states of the whole construction will be $O(2^{kn})$ where $k$ is the number of functional $\EVA$s in the expression. 

The only subtle point here is that each operation (i.e. join or union) must preserve the functional and deterministic property of the input automata in order to avoid a determinization after the join and union operations are computed. Indeed, one can easily check in the Proof of Proposition~\ref{prop:algebra-eVA} that this is the case for the join of two deterministic $\fEVA$. 
Unfortunately, the linear construction of the union of two $\fEVA$ does not preserve the deterministic property of the input automaton. For this reason, we need an alternative construction of the union that preserves determinism. This is shown in the next lemma concluding the proof of the proposition. 

\begin{lemma}
Let $\cA_1$ and $\cA_2$ be two deterministic $\fEVA$. Then there exists a deterministic $\fEVA$ $\cA_{\cup}$ such that $\cA_{\cup} \equiv \cA_1 \cup \cA_2$. Moreover, $\cA_{\cup}$ is of size $|\cA_1|\times |\cA_2|$, i.e. quadratic with respect to $\cA_1$ and $\cA_2$.
\end{lemma}
\begin{proof}
Let $\cA_1 = (Q_1, q_0^1, F_1, \delta_1)$ and $\cA_2 = (Q_2, q_0^2, F_2, \delta_2)$ be two $\fEVA$ such $Q_1\cap Q_2 = \emptyset$. The intuition behind the construction is to start running both automata in parallel, but add the possibility to branch off and continue the run in just one automaton, only when both cannot simultaneously execute a transition. 
Formally, let $\cA_{\cup}  = (Q, (q_0^1, q_0^2) , F, \delta )$ such that $Q = Q_1 \times Q_2 \cup Q_1 \cup Q_2$, $F = F_1\times Q_2 \cup Q_1\times F_2 \cup F_1 \cup F_2$, and $\delta$ satisfies that:
\begin{itemize}
\item $\delta_1 \subseteq \delta$ and $\delta_2 \subseteq \delta$, 
\item $\big( (p_1, p_2), o, (q_1, q_2) \big) \in \delta$ whenever $(p_1, o, q_1)\in\delta_1$, and $(p_2, o, q_2)\in\delta_2$,
\item $\big( (p_1, p_2), o, q_1 \big) \in \delta$ whenever $(p_1, o, q_1)\in\delta_1$, and $(p_2, o, q_2)\notin \delta_2$ for every $q_2 \in Q_2$, and
\item $\big( (p_1, p_2), o, q_2 \big) \in \delta$ whenever $(p_2, o, q_2)\in\delta_2$, and $(p_1, o, q_1)\notin \delta_1$ for every $q_1 \in Q_1$.
\end{itemize}

To show $\semd{\cA_{\cup}} \subseteq \semd{\cA_1} \cup \semd{\cA_2}$, let $\mu \in \semd{\cA_{\cup}}$ be an arbitrary mapping and $\rho$ be the corresponding run in $\cA_{\cup}$. Since $\rho$ is accepting, the last state in the run can either be from $Q_1\times Q_2$, $Q_1$ or $Q_2$. It is easy to see that either case, a run for $\mu$ exists in the original automaton. More specifically, if it is from $Q_1\times Q_2$, then both automata have complete runs defined for $\mu$, if it is from $Q_1$ or $Q_2$, then $\cA_1$ or $\cA_2$, respectively, has a defined run for $\mu$. Then, we conclude that $\mu \in \semd{\cA_1} \cup \semd{\cA_2}$.  

To show $\semd{\cA_1} \cup \semd{\cA_2} \subseteq \semd{\cA_{\cup}}$, consider $\mu$ in either $\semd{\cA_1}$  or $\semd{\cA_2}$. Without loss of generality, assume that $\mu \in \semd{\cA_1}$. Then, a run $\rho_1$ in $\cA_1$ exists that produces $\mu$. If we can also define a run $\rho_2$ of $\cA_2$ over $d$ that outputs $\mu$, then the run $\rho$ in $\cA_{\cup}$ can be constructed by coupling up states from $\rho_1$ and $\rho_2$. Since both use the same transitions, then the last state in $\rho$ must be in $F_1 \times Q_2$ and is accepting. Otherwise, $\rho_2$ cannot be defined, then some transition in $\rho_1$ is not defined in $\cA_2$. This means that $\rho$ in $\cA_{\cup}$ can be constructed by following first the transitions defined in both $\cA_1$ and $\cA_2$, but, at the first undefined transition for $\cA_2$, $\rho$ can branch off and continue on states from $Q_1$. That transition exists since it does not exist for $\cA_2$, but it does for $\cA_1$. After that, $\rho$ continues as $\rho_1$, making $\rho$ also accepting. In both cases, we conclude that $\mu \in \semd{\cA_{\cup}}$.

One can easily check that if $\cA_1$ and $\cA_2$ are functional, then $\cA_{\cup}$ is also functional, since every accepting run $\rho$ in $\cA_{\cup}$ has a corresponding accepting run either in $\cA_1$, $\cA_2$,  or both. These runs are valid and total and, thus, $\rho$ must also be valid.
From the construction, one can also check that if $\cA_1$ and $\cA_2$ are deterministic, then $\cA_{\cup}$ is also deterministic.
Finally, the size of $\cA_{\cup}$ is quadratic since it uses $O(|Q_1|\times |Q_2|)$ states and at most $O(|\delta_1|\times |\delta_2|)$ transitions. This was to be shown.
\end{proof}

%% file: appendix-spanl.tex
\subsection*{Proof of Theorem \ref{theo:counting}}

\input{counting}

\subsection*{Proof of Theorem \ref{theo:spanl}}

\input{enumeration}

%% file: counting.tex
\begin{algorithm}[t]
	%	\caption{Evaluate $A = (Q, q_0, F, \delta)$ over a word $a_1 \ldots a_n$}\label{alg:ma-eval}
	\caption{Count the number of mappings in $\semd{\cA}$ over the document  $d = a_1 \ldots a_n$}\label{alg:ma-count}
	\begin{algorithmic}[1]
		\Function{Count}{$\cA$, $a_1 \ldots a_n$}
		\ForAll{$q \in Q \setminus \{q_0\}$}
		\State $N[q] \gets 0$
		\EndFor
		\State $N[q_0] \gets 1$

		\For{$i := 1$ to $n$}
		\State $\textsc{Capturing}(i)$			
		\State $\textsc{Reading}(i)$
		
		\EndFor
		
		\State $\textsc{Capturing}(n+1)$
		
		\State \Return $\sum_{q \in F} N[q]$
		\EndFunction
		
		\medskip
		
		\Procedure{Capturing}{$i$}
		\State $N' \gets N$
		\ForAll{$q \in Q$ \textbf{with} $N'[q] > 0$}
		  	\ForAll{$S \in \markers_\delta(q)$}
				\State $p \gets \delta(q,S)$
				\State $N[p] \gets N[p] + N'[q]$
			\EndFor
		\EndFor
		\EndProcedure
		
		\medskip
		
		\Procedure{Reading}{$i$}
		
		\State $N' \gets N$
		\State $N \gets 0$
		
		\ForAll{$q \in Q$ \textbf{with} $N'[q] > 0$}
		\State $p \gets \delta(q,a_i)$
		
		\State $N[p] \gets N[p] + N'[q]$
		\EndFor
		\EndProcedure
	\end{algorithmic}
\end{algorithm}

The \textsc{Count} function in Algorithm~\ref{alg:ma-count} calculates $|\semd{\cA}|$ given a $\dsEVA$ $\cA = (Q, q_0, F, \delta)$ and a document $d = a_1 \ldots a_n$. 
This algorithm is a natural extension of Algorithm~\ref{alg:ma-eval} in Section~\ref{sec:algorithm}. 
Instead of keeping the set of list $\{\plist_q\}_{q \in Q}$ where each list $\plist_q$ succinctly encodes all mappings of runs which end in state $q$, we keep an array $N$ where $N[q]$ stores the number of runs that end in state $q$. 
Since $\cA$ is sequential (i.e. every partial run encodes a valid partial mapping) and deterministic (i.e. each partial run encodes a different partial mapping), we know that the number of runs ending in state $q$ is equal to the number of valid partial mappings in state $q$. 
Therefore, if $N[q]$ stores the number of runs at state $q$, then the sum of all values $N[q]$ for every state $q \in F$ is equal to the number of mappings that are output at the final states. 

As we said, Algorithm~\ref{alg:ma-count} is very similar to the constant delay algorithm. At the beginning (i.e. lines 2-4), the array $N$ is initialize with $N[q] = 0$ for every $q \neq q_0$ and $N[q_0] = 1$, namely, the only partial run before reading or capturing any variable is the run $q_0$. 
Next, the algorithm iterates over all letters in the document, alternating between \textsc{Capturing} and \textsc{Reading} procedures (lines 5-8).
The purpose of the $\textsc{Capturing}(i)$ procedure is to extend runs by using extended variable transitions between letters $a_{i-1}$ and $a_{i}$. 
This procedure first makes a copy of $N$ into $N'$ (i.e. $N'$ will store the number of runs in each state before capturing) and then adds to $N[p]$ the number of runs that reach $q$ before capturing (i.e. $N'[q]$) whenever there exists a transition $(p, S, q) \in \delta$ for some $S \in \markers_\delta(q)$.
On the other side, the procedure $\textsc{Reading}(i)$ is coded to extend runs by using a letter transition when reading $a_i$.
Similar to \textsc{Capturing}, \textsc{Reading} starts by making a copy of $N$ into $N'$ (line 17) and $N$ to $0$ (line 18).
Intuitively, $N'$ will store the number of valid runs before reading $a_i$ and $N$ will store the number of valid runs after reading $a_i$.
Then, \textsc{Reading} procedure iterates over all states $q$ that are reached by at least one partial run and adds $N'[q]$ to $N[p]$ whenever there exists a letter transition $(q, a_i, p) \in \delta$. Clearly, if there exists $(q, a_i, p) \in \delta$, then all runs that reach $q$ after reading $a_1 \ldots a_{i-1}$ can be extended to reach $p$ after reading $a_1 \ldots a_i$.
After reading the whole document and alternating between $\textsc{Capturing}(i)$ and $\textsc{Reading}(i)$, we extend runs by doing the last extended variable transition after reading the whole word, by calling $\textsc{Capturing}(n+1)$ in line 8. 
Finally, the output is the sum of all values $N[q]$ for every state $q \in F$, as explained before. 

The correctness of Algorithm \ref{alg:ma-count} follows by a straightforward induction over $i$. Indeed, the inductive hypothesis states that after the $i$-iteration, $N[q]$ has the number of partial runs of $\cA$ over $a_1 \ldots a_i$. 
Then, by following the same arguments as in Lemma~\ref{lemma:correctness1}, one can show that $N[q]$ store the number of partial runs of $\cA$ after capturing and reading the $(i+1)$-th letter.

%% file: enumeration.tex
Let us first define the class $\textsc{SpanL}$.
Formally, let $M$ be a non-deterministic Turing machine with output tape, where each accepting run of $M$ over an input produces an output. Given an input~$x$, we define $span_M(x)$ as the number of {\em different} outputs when running $M$ on $x$. Then, $\textsc{SpanL}$ is the counting class of all functions $f$ for which there exists a non-deterministic logarithmic-space Turing machine with output such that $f(x) = span_M(x)$ for every input $x$. We say that a function $f$ is $\textsc{SpanL}$-complete if $f \in \textsc{SpanL}$ and every function in $\textsc{SpanL}$ can be reduced into $f$ by log-space parsimonious reductions~\cite{LOGC}.

For the inclusion of $\COUNT[\fVA]$ in \textsc{Spanl}, let $M$ be a non-deterministic TM that receives $\cA$ and $d$ as input. The work of $M$ is more or less straightforward: it must simulate a run of $\cA$ over $d$ to generate a mapping $\mu \in \semd{\cA}$, and it does so by alternating between extended variable transitions and letter transitions reading $d$ and writing the corresponding run on the output tape. At all times, $M$ keeps a pointer (i.e. with log space) for the current state and a pointer to the current letter. Furthermore, it starts and ends with a variable transition as defined in Section~\ref{sec:algorithm}. Whenever a variable transition is up, the machine must choose non-deterministically from all its outgoing variable transitions from the current state. Recall that $M$ can also choose to  not take any variable transition, in which case it stays in the same state without writing on the output tape. Instead, if $(q,S, p)$ is chosen then $M$ writes the set of variables in $S$ on the output tape and updates the current state. It does so maintaining a fixed order between variables (either lexicographic or the order presented in the input). 
On the other hand, when a letter transition is up, if a transition with the corresponding letter from $d$ exists (defined by the current letter), then the current letter is printed in the output tape and, the current state and letter are updated, changing to a capturing phase.  If no transition exists from the current state, then $M$ stops and rejects. Once the last letter is read (the pointer to the current letter is equal to $|d|$), then the last variable transition is chosen. 
Finally, if the final state is  accepting, then $M$ accepts and outputs  what is on the output tape.
The correctness of $M$ (i.e. $|\semd{\cA}| \ = \ span_M(\cA,d)$) follows directly from the functional properties of $\cA$. More precisely, we know that each accepting run is valid, and will therefore produce an output. Finally, in case that $\cA$ has two runs on $x$ that produce the same output, by the definition of \textsc{Spanl} this output will be counted only once, as required to compute $|\semd{\cA}|$ correctly.

\newcommand{\cB}{\mathcal{B}}

For the lower-bound, we show that the Census problem \cite{LOGC}, which is \textsc{SpanL}-hard, can be reduced into $\COUNT[\fVA]$ via a parsimonious reduction in logarithmic-space. Formally, given a NFA $\cB$ and length $n$, the Census problem asks to count the number of words of length $n$ that are accepted by $\cB$. 
We reduce an input of the Census problem $(\cB, n)$ into $\COUNT[\fVA]$ by computing a functional $\VA$ $\cA_{\cB, n}$ and a document $d_{\cB, n}$ such that the number of words of length $n$ that $\cB$ accepts, is equivalent to count how many mappings does $\cA_{\cB, n}$ generate over $d_{\cB, n}$. 
Let $\cB = (Q, \Sigma, \Delta, q_0, F)$ be an NFA with $\Sigma = \{a,b\}$. Define $d_{\cB, n} = (\#cc)^n$ and $\cA_{\cB, n} = (Q', q_0', F', \delta')$ over the alphabet $\{c, \#\}$ such that $Q' = Q \times \{0,\ldots, n\}$, $q_0' = (q_0,0)$, $F'= F \times \{n\}$.
Furthermore, for the sake of simplification we define $\delta'$ by using extended transitions as follows:
\[
\begin{array}{rcl}
(q,a,p)\in\Delta  & \text{ then } & \Big((q,i-1), \  \# \cdot \Open{x_{i}} \cdot \, c \, \cdot \Close{x_{i}} \cdot c, \ (p,i)\Big) \in \delta' \ \ \text{for all } i\in\{1,\ldots, n\} \\
(q,b,p)\in\Delta & \text{ then } & \Big((q,i-1), \ \# \cdot c \cdot \Open{x_{i}} \cdot \, c \, \cdot \Close{x_{i}}, \ (p,i)\Big) \in \delta' \ \  \text{for all } i\in\{1,\ldots, n\}
\end{array}
\]
In the previous definition, a transition of the form $ ((q,i-1),  w, (p,i))$ means that the $\VA$ will go from state $(q,i-1)$  to the state $(p,i)$ by following the sequence of operations in $w$. 
For example the sequence $\# \cdot \Open{x_{i}} \cdot \, c \, \cdot \Close{x_{i}} \cdot c$ means that an $\#$-symbol will be read, followed by open $x_i$, read $c$,  close $x_i$, and read $c$. 
Clearly, extended transitions like above can be encoded in any standard $\VA$ by just adding more states. 

Note that to get to a state $(p,i)$ the only option is to start from the state $(q,i-1)$. Since all runs start at $(q_0, 0)$ and final states are of the form $(p,n)$, an accepting run of $\cA_{\cB,n}$ over $d_{\cB,n}$ must traverse $n + 1$ states of the form $(q,i)$, one for each $i\in\{0,\ldots, n\}$, and therefore assign all $n$ variables $x_i$. Also, between two consecutive states the transition always captures a span of length 1 (i.e. $\Open{x_{i}} \cdot \, c \, \cdot \Close{x_{i}}$) and read three characters, starting with an \#-symbol which is never captured. Therefore, all accepting runs assign all $n$ variables, and $x_i$ is either assigned to $[3i-1, 3i\rangle$ or $[3i,3i+1\rangle$. Since all the variables are opened and closed correctly between each $(q,i-1)$ and $(p,i)$, we can conclude that $\cA_{\cB,n}$ is functional.

One can easily check that the reduction of $(\cB, n)$ to $(\cA_{\cB, n}, d_{\cB, n})$ can be done with logarithmic space.
To prove that the reduction is indeed parsimonious (i.e. $|\{w \in \Sigma^n \mid w \in \cL(\cB)\}| = |\seme{\cA_{\cB, n}}{d_{\cB, n}}|$), we show that there exists a bijection between words of length $n$ accepted by $\cB$ and mappings in $\seme{\cA_{\cB, n}}{d_{\cB, n}}$. Specifically, consider the function  $f: \{w \in \Sigma^n \mid w \in \cL(\cB)\} \to \seme{\cA_{\cB, n}}{d_{\cB, n}}$ such that $f(w)$ is equivalent to the mapping $\mu_w : \{x_1,\ldots, x_n\} \to \sub(d_{\cA, n})$:
\[ 
\begin{array}{rcl}
\mu_w(x_i) & = & \begin{cases} [3i-1, 3i\rangle, & \text{if }w_i = a \\ [3i, 3i+1\rangle, & \text{if } w_i = b \end{cases}
\end{array}
\]
for every word $w = w_1 \ldots w_n \in \cL(\cB)$.
To see that $f$ is indeed a bijection, note that for every word $w \in \cL(\cB)$ of length $n$ we have an accepting run of length $n$ in $\cA$ and we can build a mapping in $\seme{\cA_{\cB, n}}{d_{\cB, n}}$. Note that all accepting runs for $w$ give the same mapping.
Moreover, note that for two different words, different mapping are defined and then $f$ is an injective function. 
In the other direction, for every mapping in $\seme{\cA_{\cB, n}}{d_{\cB, n}}$ we can build some word of length $n$ that is accepted by $\cB$ and, thus, $f$ is surjective. Therefore, $f$ is a bijection and the reduction from the Census problem into  $\COUNT[\fVA]$ is a parsimonius reduction. This completes the proof. 

% to    consider an accepting run $\rho$ of $\cB$ over $w$.
%Since there is a one-to-one correspondence between the transitions in $\cB$ and the transitions in $\cA_{\cB, n}$, we can build a run $\rho'$ from $\rho$ by using the corresponding transitions on $\cA_{\cA, n}$ starting from $(q_0,0)$. The run of $\cB$ over $w$ has length $n+1$ (i.e. the number of states) and then the final state of $\rho'$  must be of the form $(q,n)$. Therefore, $\rho'$ will be accepting for $\cA_{\cA, n}$ over $d_{\cB, n}$. We conclude that $\mu^{\rho'}$ is well defined, $\mu^{\rho'} = \mu_w$ and thus $f$ is a bijection.

%The latter function is clearly  a bijection, and therefore, Census($\cA, 1^n$) = $\text{\#MA}_{\text{ext}}^{\text{valid}}(\cA_{\cA, n}, d_{\cA, n})$. The computation of $\cA_{\cA, n}, d_{\cA, n}$ can be done in logarithmic space using a constant number of pointers to the original input $\cA$'s states and $n$. Therefore, $\text{\#MA}_{\text{ext}}^{\text{valid}}$  is span-L-hard and span-L-complete.